\documentclass{llncs}
\usepackage{pdfpages}
\usepackage{amsfonts}
\usepackage{graphicx}
\usepackage{rotating}
\usepackage{tikz}
\usetikzlibrary{calc,matrix,positioning,fit}

\begin{document}
\title{Detecting unknots via equational reasoning, I:\\
Exploration \thanks{The final publication is available at http://link.springer.com.}}
\author{Andrew Fish\inst{1} and  Alexei Lisitsa\inst{2}}
\institute{{$^1$ School of Computing, Engineering and Mathematics, University of Brighton}\\ {$^2$ Department of Computer Science, The University of Liverpool}
\email{{{Andrew.fish@brighton.ac.uk, A.Lisitsa@csc.liv.ac.uk}}}}
\bibliographystyle{plain}%
\maketitle

\begin{abstract}
We explore the application of automated reasoning techniques to unknot detection, a classical problem of computational topology. We adopt a two-pronged experimental approach, using a theorem prover to try to establish a positive result (i.e. that a knot is the unknot), whilst simultaneously using a model finder to try to establish a negative result (i.e. that the knot is not the unknot). The theorem proving approach utilises equational reasoning, whilst the model finder searches for a minimal size counter-model. We present and compare experimental data using
the involutary quandle of the knot, as well as comparing with alternative approaches, highlighting instances of interest. Furthermore, we present theoretical connections of the minimal countermodels obtained with existing knot invariants, for all prime knots of up to 10 crossings: this may be useful for developing advanced search strategies.
\end{abstract}



\section{Introduction}
One of the most well-known and intriguing problems in computational topology is \emph{unknot detection} (UKD):
given a \emph{knot}, which is a closed loop without self-intersection embedded in 3-dimensional Euclidean space $\mathbb{R}^{3}$, is it possible to
deform $\mathbb{R}^3$ continuously such that the knot is transformed into a trivial unknotted circle without passing through itself?  Knots are often studied as a diagrammatic system: (i) a knot diagram is a regular projection of the knot onto a plane, having a finite number of singularities, all of which are transverse double points annotated to indicate which strand is passing over and which is passing under at each crossing; (ii) knots are equivalent if and only if their diagrams differ by a finite sequence of Reidemesiter moves~\cite{Reid}.  Figure 1 shows the
diagrams
of two knots with a negative a) and positive b) answers to the unknottedness question. All work and results stated assume that knots are tame (a common technical requirement which is generally imposed on knots, ruling out pathological cases such as permitting infinite sequences of trefoil-like knot pieces of decreasing sizes glued together within a knot); see~\cite{lickorish} for more details, for instance.

\begin{figure}[t!]
\hspace*{3cm}\includegraphics[scale=0.15]{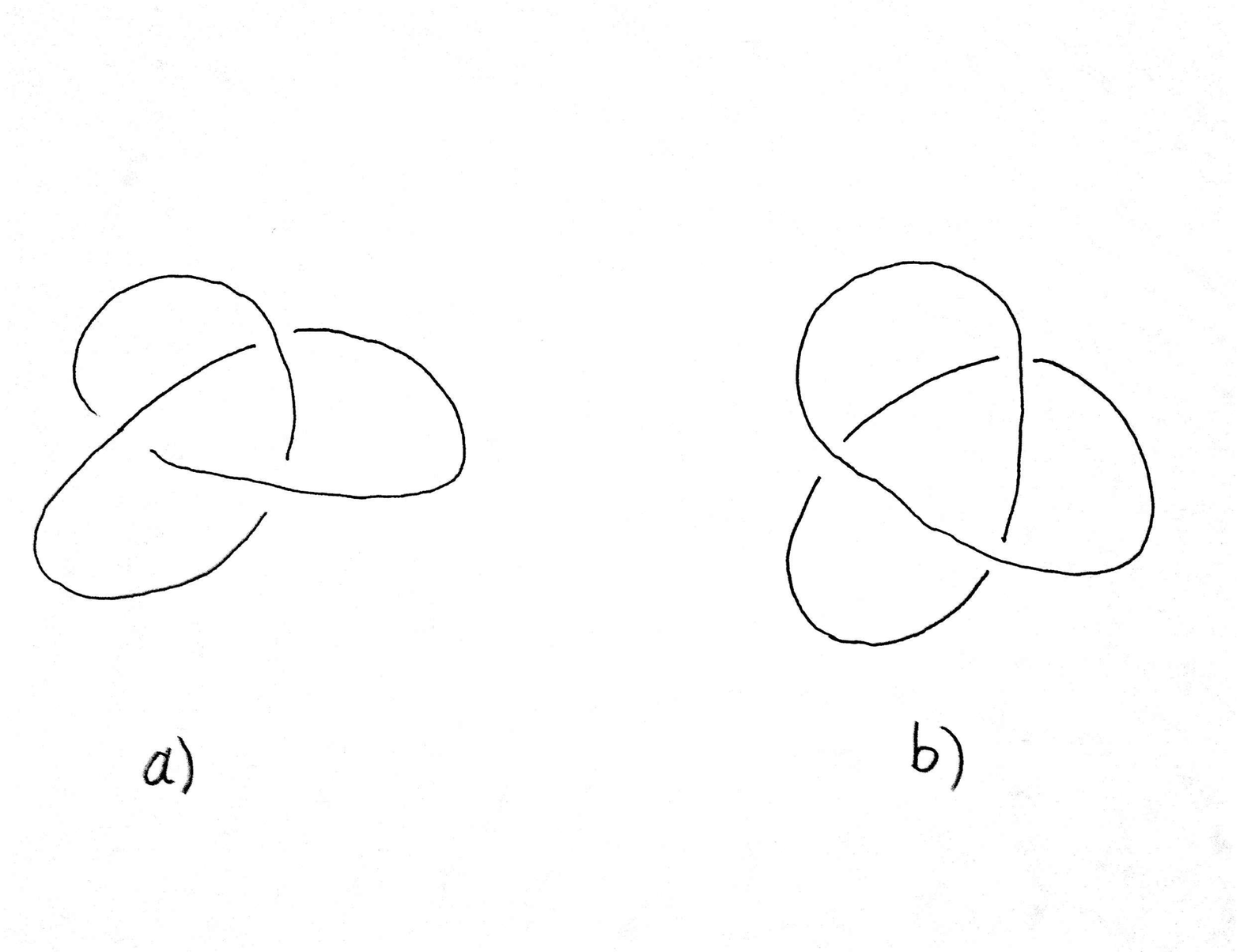}
\caption{a) non-trivial trefoil knot and b) trivial knot or unknot}
\end{figure}

The unknot detection (or unknot recognition) problem has attracted 
a lot of attention, but some of the fundamental questions about 
it still remain open. In particular, it is unknown whether it is 
possible to recognize unknots in PTIME. It is known, though, that 
the problem lies in NP $\cap$ coNP \cite{coNP,NP} (membership in $coNP$ 
is subject to generalized Riemann hypothesis holding). There has been a slow 
but steady development of algorithms for unknot detection and their experimental evaluation. An early algorithm, presented by W.~Haken  in his proof of the decidability of UKD~\cite{Hak61}, was developed for theoretical purposes and was deemed to be impractical due to being too complex to attempt to implement it. Since then various algorithms for unknot detection have been proposed with various degrees of implementability and efficiency \cite{Dyn}. The algorithms based on \emph{monotone simplifications}~\cite{Dyn} provide practically fast recognition of unknots but do not necessarily yield a decision procedure. The algorithms based on \emph{normal surface theory}, implemented in Regina system \cite{BenOZ12}, provide efficient recognition of non-trivial knots. In particular, it is reported that every non-trivial knot with crossing number $\le 12$ is recognized as such by the procedure from~\cite{BenOZ12} in under 5 minutes. There still are efficiency problems with the existing algorithms, which in the worst case are exponential, and it appears that establishing that a particular diagram with a few hundred (or even dozens of) crossings represents a non-trivial knot may well be out of reach of the available procedures.
Thus the exploration of alternative procedures for unknot detection 
is an interesting and well-justified task.


In this paper we explore the following route to the efficient practical algorithms for unknot detection. The unknotedness property can be faithfully characterized by the properties of algebraic invariants associated with knot projections. We attempt to establish the properties of concrete invariants by using methods and procedures developed in the \emph{automated reasoning} area. A key observation is that the task of unknot detection can be reduced to the task of (dis)proving a first-order formulae, and for this there are efficient generic automated procedures, notwithstanding the fact that generally first-order-order validity is undecidable.

\section{Involutory quandles and Unknot detection}

We provide relevant background definitions; for example, see~\cite{FR92,Joyce1982b,Winker,Carter} for further details.

\begin{definition}
Let $Q$ be a set equipped with a binary operation $\triangleright$ (product) such that the following hold:

\begin{description}
\item[Q1] $x \triangleright x = x$ for all $x \in Q$.
\item[Q2] For all $x,y \in Q$, there is a unique $z \in Q$ such that $x= z \triangleright y$.
\item[Q3] For all $x,y,z \in Q$, we have $(x \triangleright y) \triangleright z = (x \triangleright z) \triangleright (y \triangleright z)$.
\end{description}

Then $Q$ is called a \emph{quandle}$\!$~\footnote{A \emph{rack}~\cite{FR92} is such a $Q$ that satisfies $Q2$ and $Q3$ but not necessarily $Q1$.}. If $Q$ additionally satisfies the identity $Q2'$ below, then $Q$ is called an \emph{involutory quandle}:

\begin{description}
\item[Q2'] $(x \triangleright y) \triangleright y = x$ for all $x,y \in Q$.
\end{description}

\end{definition}

\begin{remark}
For a quandle $Q$, the unique element $z \in Q$ from axiom 2 is denoted by $z=x\triangleright^{-1}y$, and $\triangleright^{-1}$ also defines a quandle structure. However, for involutory quandles, we have $\triangleright = \triangleright^{-1}$, which can be taken as an equivalent definition of involutory; axiom 2' supersedes axiom 2.
\end{remark}


\begin{definition}
A function $\phi: Q_1 \rightarrow Q_2$ between quandles is a homomorphism if $(a \triangleright b) \phi = (a)\phi \triangleright (b)\phi$ for any $a,b \in Q_1$.
\end{definition}

%
%

Given a knot $K$ (i.e. a circle embedded in $\mathbb{R}^3$), a well known invariant is the \emph{knot group} of $K$, which is $\pi(K)=\pi_1(\mathbb{R}^3-K)$, the fundamental group of the complement of the knot $K$ in $\mathbb{R}^3$ (i.e. homotopy class of paths in the complement of the knot). One can compute a presentation of the knot group, in terms of generators and relations, from a knot diagram, following Wirtinger (e.g. see~\cite{lickorish} for details). 
An analogous construction can be used to construct a presentation of the \emph{quandle} of the knot, $Q(K)$ (e.g. see~\cite{Winker} for details).  One acquires the presentation of the knot group from the presentation of the quandle by considering the generators and relations in the group, and imposing the quandle operation to be conjugation. Since we focus primarily on involutory quandles, we provide the simplified construction for these below; a method for generalising to (not involutory) quandles is to assign an orientation to the knot, which yields a sign for each crossing according to the relative orientations of the involved curves, and then the relation assigned to the crossing is either $a \triangleright b = c$ or $a \triangleright^{-1} b = c$ according to the sign of the crossing. Interpreting the quandle relation $\triangleright$ as conjugation in the knot group (i.e.  $a \triangleright b = b^{1} a b$), and $\triangleright^{-1}$ as its inverse (i.e. $a \triangleright^{-1} b = b a b^{1}$) returns the well known Wirtinger presentation of the knot group.

\begin{definition}
A presentation of the involutory knot quandle, $IQ(K)$ for a knot $K$, is obtained from a diagram $D$ for $K$ as follows: a \emph{solid arc} of the diagram is  %
an unbroken line of the diagram with an \emph{undercrossings} at each of its ends;
%
%
every solid arc of the diagram is labelled by an unique label; all labels of $D$ form the set $G_{D}$ of \emph{generators}; to every crossing of $D$ one associates a relation, as shown in Figure~\ref{fig:crossingrelation}; denote the set of all such relations by $R_{D}$. Then the \emph{presentation} $\langle G_{D}\mid R_{D} \rangle$ defines the involutory quandle $IQ(D)$. This is a quotient of the free involutory quandle modulo the equational theory defined by $R_{D}$.
\end{definition}

The three equalities $Q1,Q2'$ and $Q3$ form an equational theory of involutory quandles, which we denote by $E_{iq}$.


\begin{figure}[t!]
\hspace{2cm}\includegraphics[scale=0.15]{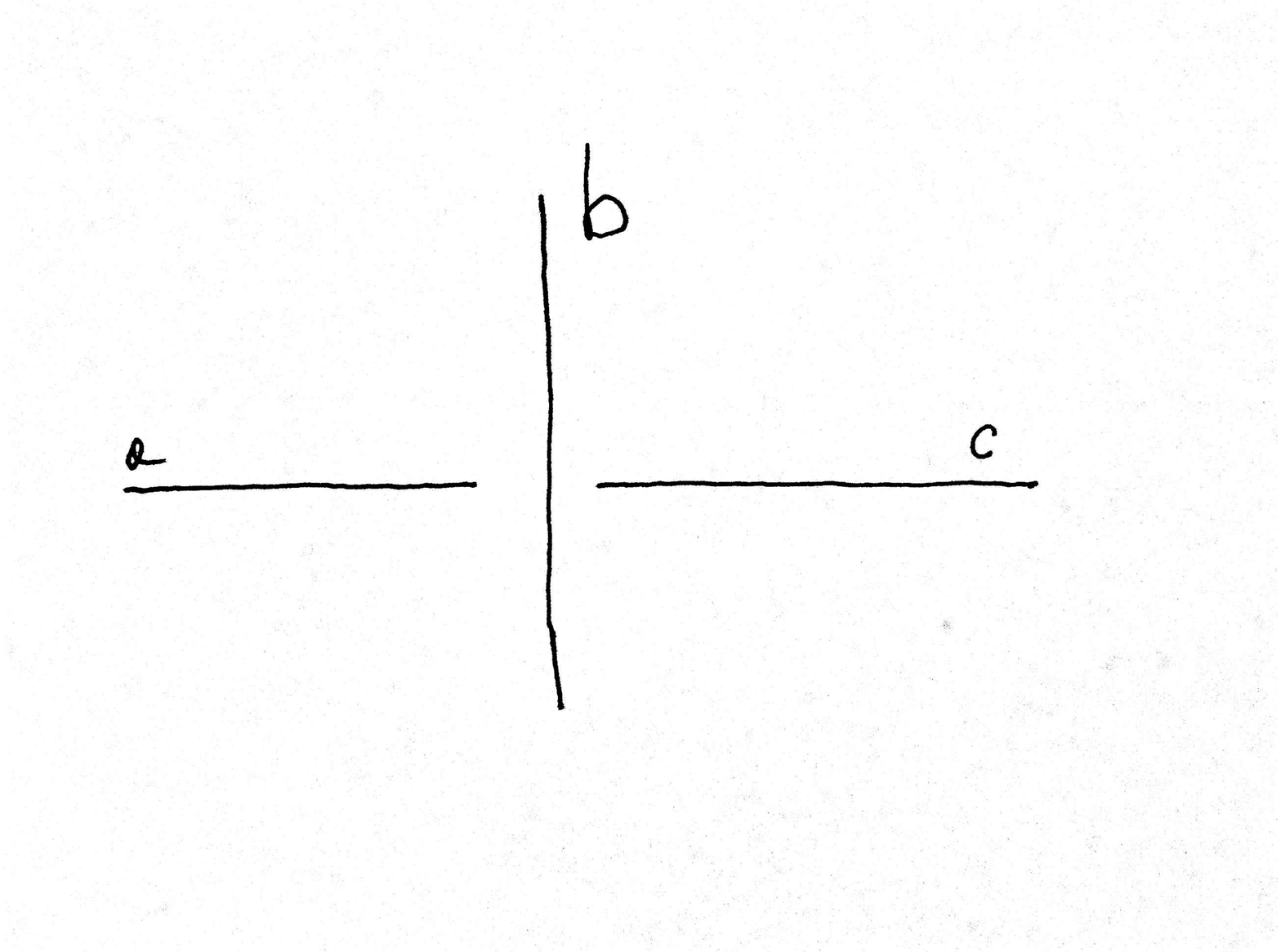} \includegraphics[scale=0.15]{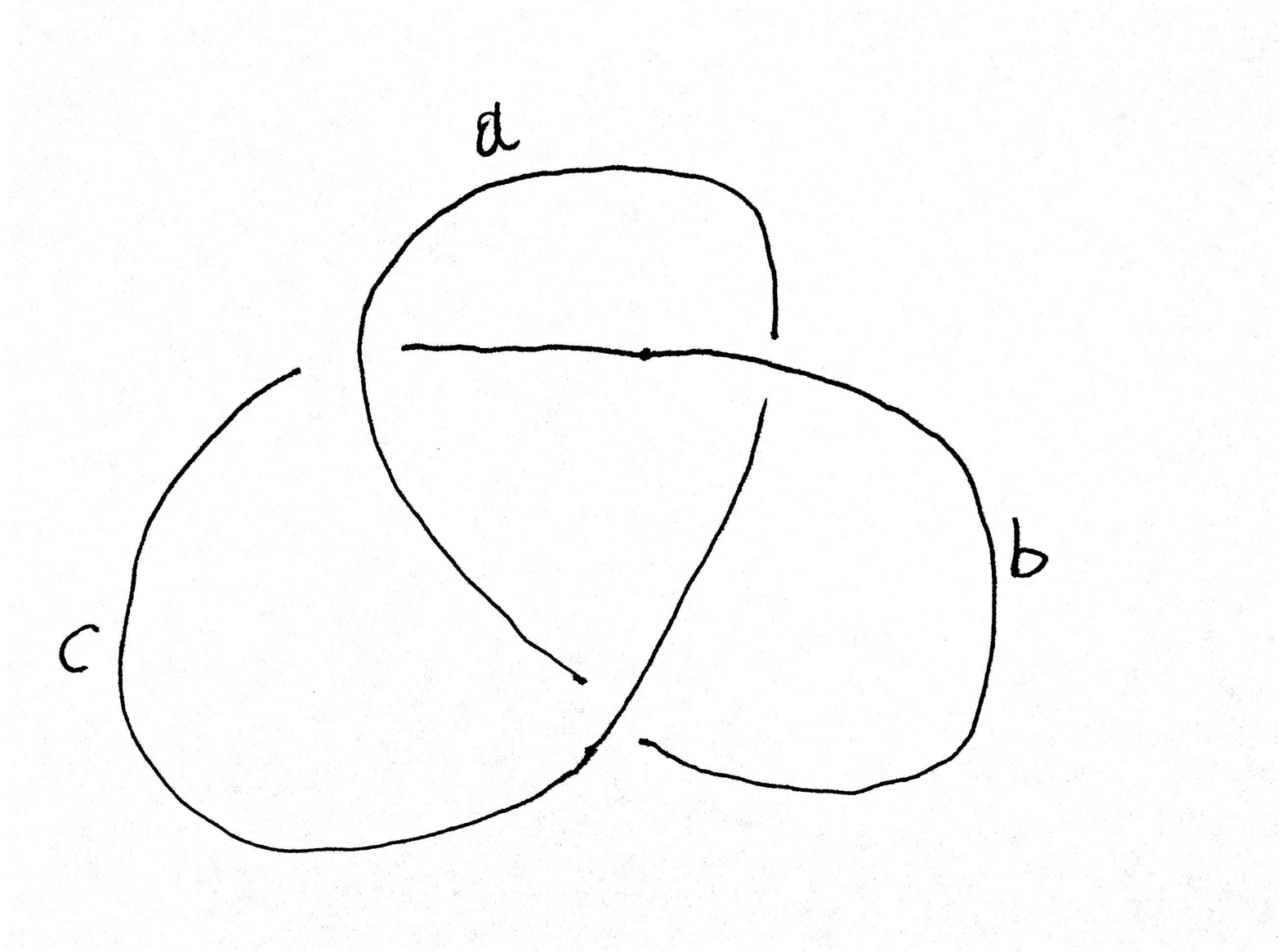}
\caption{(a) Left: A labelled crossing and its corresponding relation $a \triangleright b = c$; here $a$ and $c$ are the labels of the underarcs at this crossing, whilst $b$ is the label of the overarc, and we often identify the arcs with their labels to simplify language in discussions. \newline \hspace*{10mm} (b) Right: The trefoil knot diagram, with solid arcs $a$,$b$,$c$.}
\label{fig:crossingrelation} \label{fig:trefoil}
\end{figure}



\begin{example}
Let $D_{tr}$ be the diagram of the trefoil knot $K$ shown in Figure~\ref{fig:trefoil}. The involutory quandle of $D_{tr}$ is defined by the presentation
$IQ(D_{tr}) = \langle a,b,c \mid  a \triangleright b = c, b \triangleright c = a, c \triangleright a = b \rangle$. For comparison, the quandle $Q(D_{tr})$ has presentation
$Q(D_{tr}) = \langle a,b,c \mid  a \triangleright^{-1} b = c, b \triangleright^{-1} c = a, c \triangleright^{-1} a = b \rangle$, whilst the knot group has presentation
$G(K) = \langle a,b,c \mid  b a b^{-1} = c, c b c^{-1} = a, a c a^{-1} = b \rangle$. In detail, consider the crossing of the diagram in which $a$ and $c$ are underarcs, whilst $b$ is an overarc; i.e. match up the crossing locally with (a rotated version of) Figure 2(a). This gives rise to either the relation $a \triangleright^{-1} b = c$ or $a \triangleright b = c$, depending on whether the \emph{sign} of the crossing is negative or positive, respectively. One method for reading off the sign is to choose an orientation of the knot (i.e. pick a direction on the knot, often depicted using an arrowhead) and if $a$ is the approaching underarc, following orientation (one can traversing a knot, or part of a knot, intuitively being a walk along the around the knot along the arcs; then following orientation means that one is traversing the arc in the direction determined by the orientation), check if one turns left or right, respectively, when passing onto the overarc $b$, following orientation. In this case,
all three crossings are negative; a mirror of the diagram (i.e. exchanging all over and under crossings) would have the above presentation for knot quandle, but with $\triangleright^{-1}$ replaced by $\triangleright$, and similarly for the knot group.



%
\end{example}

\subsection{Overview of the approach}

The importance of involutory quandles, in the context of unknot detection, relies on the following properties \cite{Joyce1982a,Joyce1982b,Winker}:

\begin{itemize}
\item Involutory quandle is a knot invariant, i.e. it does not depend on the choice of diagram;  
\newline [Theorem 15.1 of~\cite{Joyce1982a} shows that the quandle $Q(K)$ of
knot $K$ is an invariant of the knot type of $K$, and the involutory quandle $IQ(K)$ is a homomorphic image of $Q(K)$]
\item Involutory quandle $IQ(K)$ of a knot $K$ is trivial (i.e. it contains a single element $e$ with $e *e = e$) if and only if $K$ is the \emph{unknot}. \newline [Theorem 5.2.5 of~\cite{Winker}].
\end{itemize}

These properties suggest the following approach to unknot detection. Given a knot diagram, one can try to decide whether its associated involutory quandle is trivial. Notice that an involutory quandle of a
knot can be an infinite set~\cite{Winker}. Not much progress has been made towards the development of specific decision procedures for such a problem, apart of that presented in the thesis of S.~Winker~\cite{Winker}; the diagrammatic method presented there, together with details and explanations, allows one to construct the involutory quandles for many knot diagrams, and in our opinion, is a very good starting point for developing algorithmic procedures directly dealing with the involutary quandles. 
In this paper, we take an alternative route and propose, instead of applying a specific involutory quandles decision procedure, to tackle unknot detection as follows:

\begin{itemize}
\item Given a knot diagram, compute its involutary quandle presentation;
\item Convert the task of involutary quandle triviality detection into the task of proving a first-order equational formula;
\item Concurrently, apply generic automated reasoning tools for first-order equational logic to tackle the (dis)proving task
\end{itemize}

Thus, we concurrently search for a proof and for a model to disprove the formula.
After explaining the details, we apply these methods in parallel, present empirical data for many knots, and
compare the countermodels with existing knot invariants that are encoded in the
smallest homomorphic image of the involutory quandle of the knot. As an overview
of the related objects, Figure~\ref{fig:overview} shows the following. In the top row, we have the knot $K$, fundamental quandle of the knot, $Q(K)$, and the projection from $Q(K)$ onto the fundamental group $\pi(K)$ obtained by forgetting information (peripheral subgroup and meridian, discussed later), and essentially setting the quandle operation to be conjunction. In the second row, the diagram $D_K$ of the knot $K$, together with presentations of $Q(K)$ and $\pi(K)$ obtained from those diagrams. The third row shows the involutory quandle presentation $IQ(K)$ obtained by identifying the quandle operation with its inverse. Here $C$ is a finite involutory quandle that is the homomorphic image of the involutory quandle of the knot. Then, if $X$ is any quandle, one can ask if imposing the involutory condition on $X$ can yield such a $C$.

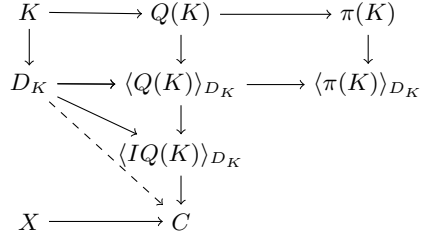
\begin{figure}[t!]
    \centering
    \begin{tikzpicture}
      \matrix (m) [matrix of math nodes,column sep=0.7cm, row sep=0.4cm]
  {
   K & Q(K) & \pi(K) \\
   D_K & \langle Q(K) \rangle_{D_K} & \langle \pi(K) \rangle_{D_K}  \\
     & \langle IQ(K) \rangle_{D_K} &  \\
    X & C &  \\
  };

\foreach \x in {1,2} {
   \draw (m-\x-1) edge[->] (m-\x-2) (m-\x-2) edge[->] (m-\x-3);

   }
\foreach \x in {2} {
   \draw (m-\x-1) edge[->] (m-\x-2);
   }

 \foreach \y in {1,2,3} {
      \draw (m-1-\y) edge [->] (m-2-\y);
     }

   \foreach \y in {2} {
      \draw (m-2-\y) edge [->] (m-3-\y);
     }

    \foreach \y in {2} {
      \draw (m-3-\y) edge [->] (m-4-\y);
     }

\draw (m-4-1) edge[->] (m-4-2); 

\draw (m-2-1) edge[dashed,->] (m-4-2);

\draw (m-2-1) edge[->] (m-3-2);

    \end{tikzpicture}

\caption{An overview of the objects related to the unknot detection programme.}
\label{fig:overview}
\end{figure}


\subsection{Unknot detection by equational reasoning}~\label{sec:untangle}

Given a knot diagram $D$, with $n$ arcs, consider its involutory quandle representation $IQ(D) = \langle G_{D}  \mid R_{D} \rangle$ with $G_{D} = \{a_{1}, \ldots, a_{n} \}$. Denote by $E_{iq}(D)$ an equational theory of $IQ(D)$, i.e.  $E_{iq}(D) = E_{iq} \cup R_{D}$.  It is
known that the axioms of (involutory) quandles are algebraic 
counterparts of the Reidemeister moves (see further discussion of 
that in Section~\ref{sec:untangle}).

\begin{proposition}\label{prop:iq-detection}
A knot diagram $D$ is a diagram of the unknot if and only if $E_{iq}(D) \vdash \wedge_{i=1 \ldots
n-1} (a_{i} = a_{i+1})$, where $\vdash$ denotes derivability in the equational logic (or, equivalently in the first-order logic with equality).
\end{proposition}
{\bf Proof. }{\em (Sketch)} $D$ is a diagram of unknot iff $IQ(D)$ is a trivial involutive quandle \cite{Winker}.
The proposition  ``$IQ(D)$ is a trivial involutory quandle iff  $E_{iq}(D) \vdash \wedge_{i=1 \ldots
n-1} (a_{i} = a_{i+1})$" is an easy consequence of the soundness and  completeness of equational logic (Birkhoff Theorem)~\cite{Bir35}. See also Lemma 4.2.7 p. 30 of \cite{Winker}. The case of first-order logic with equality follows from the conservativity of first-order logic with equality over equational logic for equational theories. \hfill{\qed}

~ \newline
Now,
if  $E_{iq}(D) \vdash \wedge_{i=1 \ldots
n-1} (a_{i} = a_{i+1})$ holds true, then this fact can be established by a proof of the formula  $E_{iq}(D) \rightarrow  \wedge_{i=1 \ldots
n-1} (a_{i} = a_{i+1})$ by a \emph{complete} automated theorem prover for first-order logic with equality, of which there are
many around, see e.g.~\cite{ATP}.
By a complete theorem prover we mean an automated procedure which, given a valid formula, terminates with a proof of the formula.

For an introduction to automated theorem proving see e.g.~\cite{GLM}.
In order to show that $E_{iq}(D) \vdash \wedge_{i=1 \ldots
n-1} (a_{i} = a_{i+1})$ does not hold, it suffices to disprove $E_{iq}(D) \rightarrow  \wedge_{i=1 \ldots
n-1} (a_{i} = a_{i+1})$. We propose to do this by the application of generic finite model finding procedures \cite{Model,McCune} to find a finite countermodel to the formula, or equivalently a finite model for $E_{iq}(D) \land  \neg\wedge_{i=1 \ldots
n-1} (a_{i} = a_{i+1})$.
So, the unknot detection procedure $P$ which we propose here consists of the parallel composition of
\begin{itemize}
\item automated proving  $E_{iq}(D) \rightarrow  \wedge_{i=1 \ldots
n-1} (a_{i} = a_{i+1})$, and
\item  automated disproving  $E_{iq}(D) \rightarrow  \wedge_{i=1 \ldots
n-1} (a_{i} = a_{i+1})$ by a finite model finder.
\end{itemize}

It is obvious that the parallel composition above provides with \emph{at least} a semi-decision algorithm for unknotedeness.
If $D$  is a diagram of the unknot then the termination of the theorem proving is guaranteed by the completeness of a theorem prover. On the other hand, if $D$ is a diagram of a non-trivial knot then the termination can be guaranteed only if a \emph{finite} countermodel exists. In general, in the first-order logic,  there are formulae which can only be refuted on \emph{infinite} countermodels, so for arbitrary formulae the termination of the automated disproving cannot be guaranteed.

For the specific type of formulae $E_{iq}(D) \rightarrow  \wedge_{i=1 \ldots n-1} (a_{i} = a_{i+1})$  we conjecture that they have \emph{finite countermodel property}, that is if there exists a countermodel for a formula of this form at all, then there is a finite countermodel too. Since a countermodel for $E_{iq}(D) \rightarrow  \wedge_{i=1 \ldots
n-1} (a_{i} = a_{i+1})$ is a model for  $E_{iq}(D) \land \neg \wedge_{i=1 \ldots
n-1} (a_{i} = a_{i+1})$, it follows that: 1) such a countermodel is a \emph{homomorphic} image of the involutory quandle $IQ(D)$ of $D$, by satisfaction of $E_{iq}(D)$; and 2) it is non-trivial involutory quandle, by satisfaction of $\neg \wedge_{i=1 \ldots n-1} (a_{i} = a_{i+1})$.

Thus the required \emph{finite countermodel} property, for the programme to yield a decision procedure, is equivalent to the property of the involutory quandles of knots being \emph{residually finite}, as formulated in the following conjecture.

\begin{conjecture}[Involutory quandles are finitely residual] ~\label{conj:resfinite}
For any knot diagram $D$, if $IQ(D)$ is not trivial (i.e. consists of more than 1 element), then there is a finite non-trivial involutory quandle $Q$ which is a homomorphic image of $IQ(D)$.
\end{conjecture}

We remark on the conjecture. Following Hempel and Thurston, we know that knot groups are residually finite. Thurston~\cite{Thurston} states that ``It is a standard fact that a finitely generated subgroup of $GL_n(\mathbb{Q})$ (the general linear group with coefficients in the rationals) is residually finite. Using this, one easily sees that the fundamental group of any geometric 3-manifold is residually finite. After a certain amount of fussing, one can assemble finite quotients of the fundamental groups of pieces of a geometric decomposition of a 3-manifold to obtain finite quotients of the fundamental group of the entire manifold.". The question of whether the proof can be lifted to quandles is another matter. Since the knot quandle contains the same information as a group system which is the triple $(G,P,m)$ consisting of the knot group $G$, a peripheral subgroup $P$, and a
meridian $m$ in $P$, one would require that the \emph{group system} (see e.g. \cite{Winker,Joyce1982a}) is somehow preserved, and the finite homomorphic image has an induced quandle structure.

\begin{theorem}
The unknot detection procedure, $P$, given above, is a decision procedure if conjecture~\ref{conj:resfinite} holds and a semi-decision procedure otherwise.
\end{theorem}

In the next section we illustrate the practical applicability of the proposed (semi)-decision procedure to various instances of unknot detection problem. In the experiments, we use an automated theorem prover Prover9 and a finite model finder Mace4, both by W.~McCune~\cite{McCune}. We present some examples of the results of the approach on unknots with interesting properties.

\section{Experiments: detecting unknots}

\subsubsection{Culprit Unknot} is shown in Figure~\ref{fig:culprit}. This is an interesting unknot which, during any untangling of it by Reidemeister moves, necessarily requires an increase in the number of crossings. 
The formula of the form $E_{iq}(D) \rightarrow  \wedge_{i=1 \ldots n-1} (a_{i} = a_{i+1})$ for the culprit unknot diagram in the syntax of Prover9/Mace4 (with $\ast$ denoting involutory quandle operation $\triangleright$) is presented to the right of the figure. Prover9 proves the formula in 0.03 seconds demonstrating thereby that culprit is indeed the unknot. 
The entire proof can be found in Appendix
%


\begin{figure}[t!]
\centering
\begin{minipage}[l]{.3\textwidth}
\centering
    \includegraphics[scale=0.15]{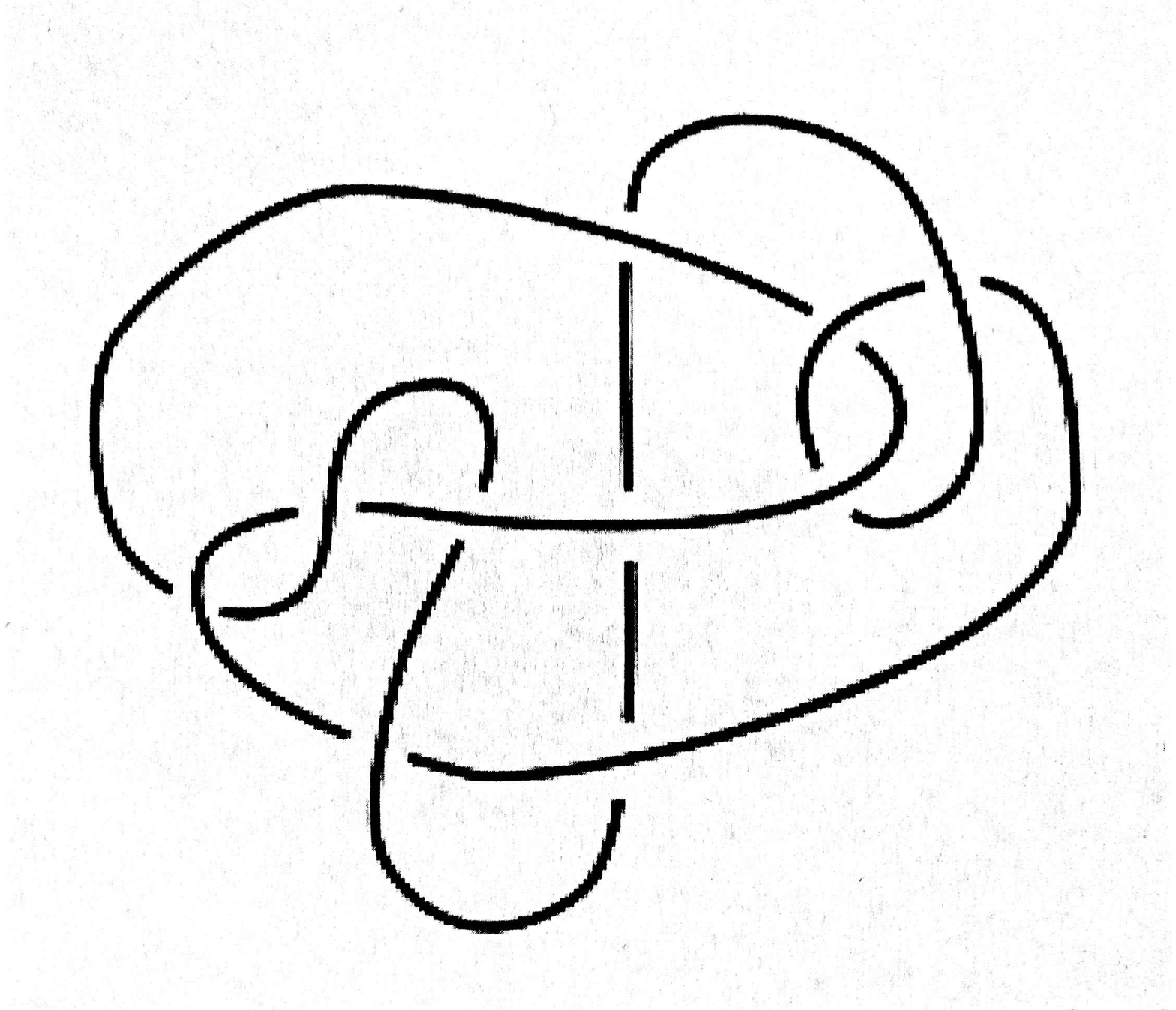}
    \caption{Culprit Unknot}
    \label{fig:culprit}
\end{minipage}
\hspace{.05\textwidth}
\begin{minipage}[c]{.3\textwidth}
\centering

Assumptions:
\begin{verbatim}
%Involutory quandle axioms
x * x = x.
(x * y) * y = x.
(x * z) * (y * z) = (x * y) * z.
%Culprit unknot
a1 = a9 * a7.
a3 = a1 * a2.
a2 = a3 * a4.
a5 = a2 * a10.
a6 = a5 * a4.
a7 = a6 * a1.
a8 = a7 * a4.
a10 = a8 * a9.
a4 = a10 * a3.
a9 = a4 * a8.
\end{verbatim}
\end{minipage}
\begin{minipage}[r]{.3\textwidth}
\centering
\vspace{5em}
Goals:
\begin{verbatim}
(a1 = a2) & (a2 = a3) &
(a3 = a4) & (a4 = a5) &
(a5 = a6) & (a6 = a7) &
(a7 = a8) & (a8 = a9) &
(a9 = a10).
\end{verbatim}
\end{minipage}

\end{figure}

%


\subsubsection{Haken's Gordian unknot} diagram has 141 crossings, and is the one of the most well-known concrete, hard-to-detect, unknots; see Figure~\ref{fig:haken}. Prover9 produces the proof of the formula of the form $E_{iq}(D) \rightarrow  \wedge_{i=1 \ldots
n-1} (a_{i} = a_{i+1})$ for this diagram in just under 15 seconds, demonstrating that indeed it is the unknot. The input, and the proof produced by the prover, can be found in~\cite{LogKnot}.

The only alternative approach capable of detecting unknotedness of Haken's Gordian Unknot in practice, that we are aware of, is Dynnikov's
  algorithm based on \emph{monotone simplifications}~\cite{Dyn,Dyn1,KnotSimplifier}.

We have experimented also with the detection of other well-known  hard unknots, such as Goerlitz unknot,
Thistlethwaite unknot, Friedman's  Twisted unknot. In all of these cases Prover9 was able to establish unknotedness
in under a second. See further details in \cite{LogKnot}. This can be compared with the times ``in only a few seconds'' required to detect unknotedness
of these instances by the Heegaard tool, reported by the author in~\cite{simplification}.


\begin{figure}[t!]
\begin{center}
\includegraphics[scale=0.25]{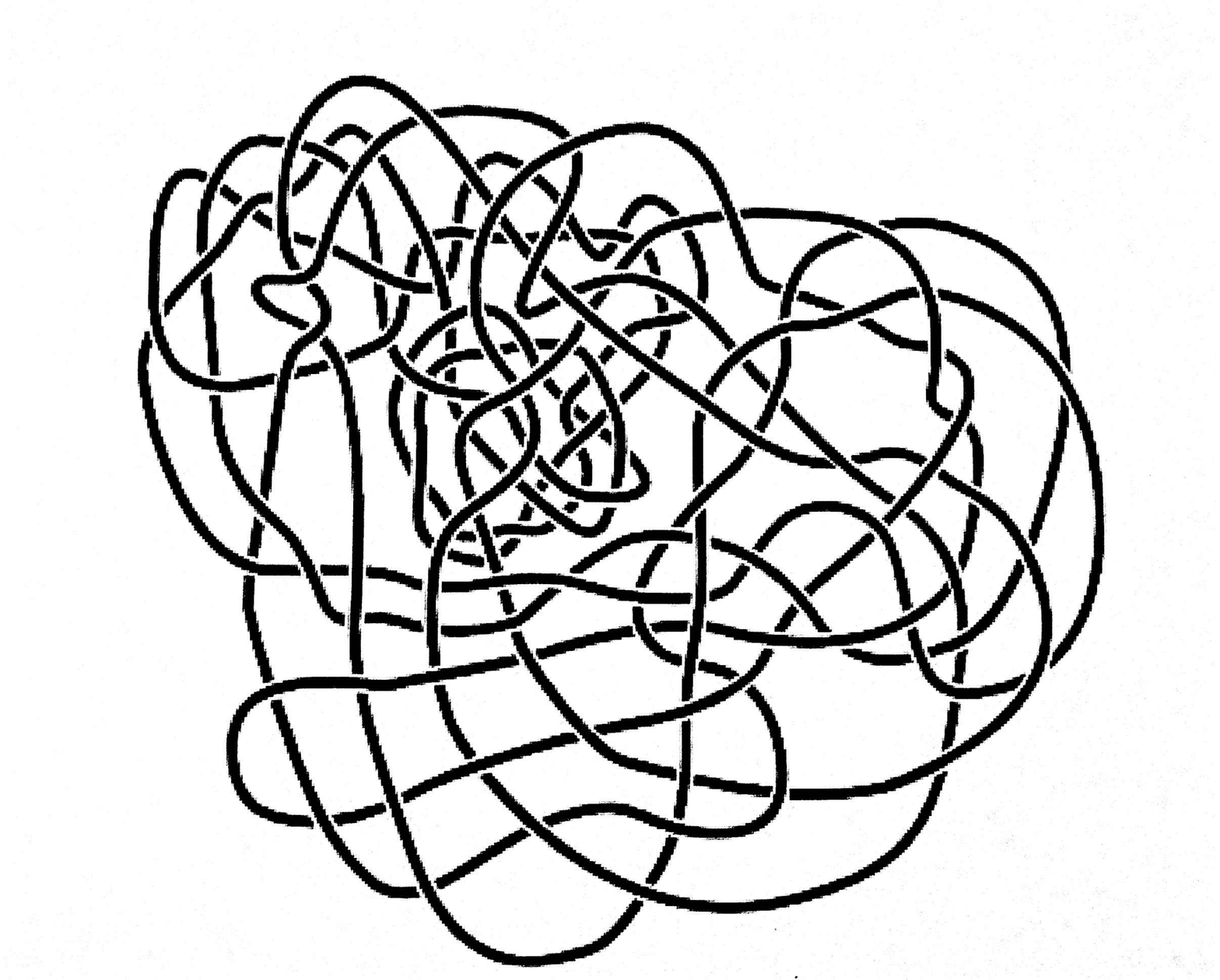}
\caption{Haken's Gordian Unknot}
\end{center}
\label{fig:haken}

\end{figure}


\section{Experiments: detecting non-trivial knots}
We first provide an example of the output from the counter-model finder for a non-trivial knot. Next, we present a table containing the time taken for each of the prime knots in the knot tables, for up to 10 crossings.
We will compare the output sizes with the known invariant of knots, called the determinant of the knot, later.


\subsubsection{Trefoil} (Figure 1 a) is the simplest non-trivial knot; the countermodel found is:

\begin{verbatim}
interpretation( 3, [number=1, seconds=0], [
        function(a1, [ 0 ]),
        function(a2, [ 1 ]),
        function(a3, [ 2 ]),
        function(*(_,_), [
			   0, 2, 1,
			   2, 1, 0,
			   1, 0, 2 ])
]).
\end{verbatim}

The table of prime knots, the size of the minimal countermodel found, and the time taken to find,
%
is given
in the Appendix.

\subsubsection{Comparisons}


 For the collection of prime non-trivial knots of up to 12 crossings from \cite{knotinfo}, the system in~\cite{BenOZ12}, based on linear programming in conjunction with normal surface theory, claims to massively improve on previous approaches, claiming to solve all cases efficiently (in under 5 minutes). Their time-based data is subsequent to highly optimised polynomial time pre-processing simplifications and uses an encoding based on triangulations of the complement of the knot. Bearing in mind the differences in encodings we have compared the performance of the Regina unknot detection algorithm with our approach using Mace4 on the knots with 10 crossings.  For the five special cases  $10_{83},10_{91},10_{92},10_{117},10_{119}$  for which our approach does not terminate in a reasonable time, whilst Regina completed the work in under 3 minutes per knot. 
For remaining   cases the average time was  47s  for Regina, and  1230s for our approach. In general our approach demonstrates much higher discrepancy in timing data: it was very efficient for the cases when small countermodels were found. For countermodels sizes up to 15-17 the 
detection time is under a second -- that holds in  more than $70 \%$ of instances, where our approach  outperforms  Regina's algorithm. In a few cases with large countermodels (e.g $10_{88},10_{94},10_{115}$) 
it takes 40000-80000s to complete the search. Further comparisons  on the large sets of knots is a subject of ongoing and future work.

\section{Countermodels and knot invariants}

Any countermodel found is a finite quandle, $C$, which is a homomorphic image of $IQ(K)$. Thus it a homomorphic image of $Q(K)$ which factors through $IQ(K)$. Such finite quandles may be constructed via ``involutising" quandles which are homomorphic images of $Q(K)$. The countermodel search process finds the smallest such finite quandle $C$, and we see that these do not all arise via the same known quandles. We present results which demonstrate that the majority, but not all, of the small alternating knots (of size up to 10) arise as quotients of the dihedral quandle. The question of how the other minimal size finite quandles arise is still open, but this demonstrates that the methodology is particularly interesting in that it is discovering the smallest size quandle invariant, over all such invariants, for each case, as opposed to computing each invariant in turn, as per the common approach using invariants.


\begin{definition}[dihedral quandle]
Let $R_n$ be the set of reflections in the dihedral group $D_{2n}$ of order $2n$ (which one can regard as the symmetry group of the regular $n$-gon). Then $R_n$ forms a quandle of order $n$,
called the dihedral quandle of order $n$.
\end{definition}

\begin{proposition}
For any knot $K$, with determinant not equal to ${\pm 1}$, $R_p$ is a finite non-trivial involutory quandle which is the image of the fundamental quandle of the knot, where $p$ is smallest prime divisor of the determinant of the knot.
\end{proposition}

\begin{proof}
For more details on racks and quandles, and this construction, see~\cite{FR92}. A homomorphic image of the fundamental quandle of the knot $K$ into $R_n$ may be given by colouring the arcs of any diagram of $K$ with $n$ colours $0,1,\ldots, n-1$ such that at each crossing if $x_{a}$, $x_b$, $x_c$ are the three colours assigned to the arcs labelled $a$, $b$, $c$, with $b$ the overarc, then $x_c \equiv 2 x_b - x_a$ mod $n$. If $n$ is prime, then it is well known that these equations have a non-constant solution if and only if $n$ divides $det(K)$ the \emph{determinant} of $K$ (obtainable as the evaluation of the Alexander polynomial at $t=-1$, sometimes denoted $\Delta(-1)$). In general, a representation into any finite quandle can be interpreted as a suitable colouring scheme for the diagram. Since the elements of the dihedral quandle are all reflections, the quandle is involutory by definition.
\end{proof}

The experimental computation, together with comparison of the determinant of the knot gives us the following, where as usual, the numbering convention is that generally adopted for prime knots in the knot tables; see~\cite{knotinfo} for example.


\begin{proposition}
Out of the 251 prime, alternating knots of up to 10 crossings, from the knot tables, a smallest non-trivial involutory quandle which is a homomorphic image of the fundamental quandle of the knot is: of size 15 for 22 knots: $9_{22}, 9_{25}, 9_{30}, 9_{36}, 9_{44}, 9_{45}, 10_{46}, 10_{47}, 10_{49}, 10_{70}, 10_{72}, 10_{73}, 10_{79}, 10_{80}, 10_{93}, 10_{102}, \newline 10_{124}, 10_{126}, 10_{127}, 10_{148},10_{149},10_{153}$; of size 28 for 11 knots: $10_{50},10_{51},10_{52},10_{53}, \newline 10_{54},10_{55},10_{57}, 10_{131}, 10_{135}, 10_{150}, 10_{151}$; of size 31 for 1 knot: $10_{115}$; of size 32 for 1 knot: $10_{118}$; of size 36 for 1 knot: $10_{110}$; of size equal to the smallest prime divisor of the determinant of the knot for the remaining $213$ knots.
\end{proposition}

\begin{corollary}
For 213 of the 251 prime, alternating knots of up to 10 crossings, there is no smaller non-trivial involutory quandle which is a homomorphic image of the fundamental quandle of the knot than the dihedral rack on $p$ elements, where $p$ is smallest prime divisor of the determinant of the knot. For the remaining 38 knots, there is a smaller such non-trivial involutory quandle.
\end{corollary}


\subsection{Discussion: Countermodels and Small Quandles}

Winker~\cite{Winker} remarks that for a certain class of knots (technically, those which are the closure of 4-strand braids), the involutory quandle $IQ(K)$ is finite and has order equal to the knot determinant $|K|$ (or $det(K)$), citing~\cite{Joyce1982a}, and that since every prime knot of 7 or fewer crossings is 4-strand it has finite involutory quandle. On the other hand, $IQ(8_{16})$ and $IQ(9_{35})$ are infinite, as are the $(k,m,n)$-pretzel knots (knots obtainable by a construction involving a certain process of $k$ twists, $m$ twists and $n$ twists) when $1/k+1/m+1/n \leq 1$; knots $8_{5} = K_{2,3,3}$ and  $9_{35} = K_{3,3,3}$ are examples of such pretzel knots. 
There does not appear to be an existing complete classification of involutory quandles. We observe that for any knot $K$ which has a finite involutory quandle $IQ(K)$, this involutory quandle itself would be a countermodel, as would the projection onto the quandle arising from each of the colouring numbers. Whilst we we may find the homomophic image of the quandle with size the smallest prime divisor of $det(K)$, corresponding to the smallest colouring number, this does not, a priori, rule out smaller homomorphic images of involutory quandles that arise in other ways. Furthermore, there are some prime knots $K$ with $det(K)= \pm 1$. The knots $10_{124}$ and $10_{153}$ are the only such prime alternating knots with up to 10 crossings. For both of these knots we find a smallest homomorphic image of involutory quandle of size 15. In~\cite{FR92}, they observe that the representation in a reflection rack, whose elements are the edges of a dodecahedron, can be used to distinguish knots which have determinant ${\pm 1}$ and so have no non trivial representation to $R_n$, giving $10_{124}$ an an example. Similarly, in~\cite{Joyce1982a}, they present the knot $10_{124}$ with determinant 1, indicating that $AbQ_2(K)$ is trivial (this is the abelian, involutory quandle on $K$, where abelian means that $(w \triangleright x) \triangleright (y \triangleright z)= (w \triangleright y) \triangleright (x \triangleright z)$).  But the involutory quandle $IQ(K)=Q_2(K)$ is non-trivial and has order 30, and may be faithfully represented on a sphere as the 30 midpoints of the edges of a dodecahedron projected onto the sphere.

In~\cite{computable}, they use the Library for Automated Deduction Research in the endeavour of identifying isomorphism classes of small quandles, and they present several families of quandles; as well as considering the dihedral quandle, they also refer to linear quandles, the Alexander quandles, and transposition quandles. These are candidate classes to consider in the identification of the countermodels. For instance, the transposition quandle, $T_n$, has size $n(n-1)/2$, and so these are candidates to explain our countermodels arising at sizes 15, 28 and 36. The identification and understanding of exactly which quandle families have smallest homomorphic image is an intriguing open problem to be explored in future work, with guidance from the countermodel finder approach adopted. Furthermore, Proposition 11.2 of~\cite{Joyce1982a} says that every involutory quandle is representable as an involutory quandle with geodesics, and so the construction of the smallest such involutory quandle with geodesics for a knot $K$ will, in fact, correspond to our search for minimal countermodel.

Winker~\cite{Winker} states that the involutory quandle of a knot or link is either finite or ``not too infinite", and gives examples to show that knots can have different knot groups but the same involutory quandle (e.g. the Figure of Eight knot and the $(5,2)$ torus knot), and that the involutory quandle of a particular prime link (the Borromean rings) is infinite.


\section{Equational Reasoning and Untangling Unknots}~\label{sec:untangle}

Recall Proposition~\ref{prop:iq-detection}: a knot diagram $D$ is a diagram of the unknot if and only if
$E_{iq}(D) \vdash \wedge_{i=1 \ldots n-1} (a_{i} = a_{i+1})$, where $\vdash$ denotes derivability in the equational logic (or, equivalently in the first-order logic with equality). We adopt the abbreviation $TRIV \equiv \wedge_{i=1 \ldots
n-1} (a_{i} = a_{i+1})$ for the generators $a_{1}, \ldots, a_{n}$. Then the condition above will be rewritten as
$E_{iq} \vdash ? TRIV$.
The axioms of involutory quandles can be seen as  algebraic counterparts of the Reidemeister moves:

\begin{enumerate}
\item $x \triangleright x = x$ for all $x \in Q$     ($\sim RM_{1}$)
\item $(x \triangleright y) \triangleright y = x$ for all $x,y \in Q$   ($\sim RM_{2}$)
\item $(x \triangleright z) \triangleright (y \triangleright z) = (x \triangleright y) \triangleright z$  for all $x,y,z \in Q$ ($\sim RM_{3}$)
\end{enumerate}

For $I \subseteq \{1,2,3\}$, denote by $E_{iq}^{I}$ an equational theory formed by the corresponding subset of the axioms $1-3$ given above. In particular $E_{iq}^{\{1,2,3\}} = E_{iq}$. 
Reidemeister's theorem~\cite{Reid} says that a diagram $D$ is a diagram of an unknot if and only if $D$ can be transformed to a trivial diagram
%
$D_{U}$
by a finite sequence of Reidemeister moves. Denote by $D \rightarrow^{I} D'$ the fact that $D$ can be transformed to $D'$ using the Reidemeister moves drawn only from $I$.
In this section we explore possible connections between equational proofs and Reidemeister transformations. The following proposition expresses the fact that the equational proof can simulate simplifications by Reidemeister moves.

\begin{proposition}~\label{prop:simulation}
For any non-empty $I \subseteq \{1,2,3\}$, if $D \rightarrow^{I} D^{U}$ then $E_{iq}^{I} \vdash TRIV$.
Furthermore an equational proof can be constructively built by a simple procedure from the untangling sequence of Reidemeister moves.
\end{proposition}

{\bf Proof.} (Sketch) Consider the case of $I = \{1,2,3\}$. Assume that for a diagram $D$ we have $D \rightarrow^{I} D^{U}$. That is, there is a sequence of diagrams $D = D_{1}, \ldots, D_{i}, \ldots, D_{n} = D^{U}$ such that every diagram in the sequence is obtained from the
previous one by a single application of a Reidemeister move. Let $IQ(D) = \langle G_{D} \mid R_{D} \rangle$ be a presentation of the involutary quandle of $D$ with the set of generators $G_{D}$ and set of relators $R_{D}$. Denote by ${\cal T}(D)$ the set of all terms built upon the set of constants, identified with the generators $G_{D}$, together with the involutory quandle operation $\triangleright$ as the only term construct. Denote by $A(D_{i})$ the set of solid arcs of the diagram $D_{i}$. A \emph{labelling} $L$ of a diagram $D_{i}$ is a mapping  $L: A(D_{i}) \rightarrow {\cal T}(D)$.  Now we demonstrate the inductive construction of a sequence of pairs $(E(D_{i}),L_{i})$, associating with each $D_{i}$ a set of equations $E(D_{i})$ and a labelling $L_{i}: A(D_{i}) \rightarrow {\cal T}(D)$ satisfying the following properties:

\begin{enumerate}
\item $E(D_{i}) \subseteq E(D_{i+1})$;
\item$E(D_{i}) \vdash E(D_{i+1})$;
\item $L_{i}$ is consistent w.r.t. involutory quandle labelling rules on solid arcs of $A(D_{i})$, meaning $E(D_{i}) \vdash L_{i}(a) \triangleright L_{i}(b) = L_{i}(c)$ for all $a,b,c \in A(D_{i})$ positioned as shown in Figure 2 (a);
\item If $t \in {\cal T}(D)$ is in $L_i(A(D_{i}))$ but not in $L_{i+1}(A(D_{i+1})$), and $\mid \! A(D_{i})\! \mid>1$ then there exist an $s$  in $L_{i+1}(A(D_{i+1}))$ and $t=s$ in $E(D_{i+1})$.
 \end{enumerate}


The intuition is that the equation set grows, adding statements of equality which are derived from the axioms, according to the progress in the unknotting sequence.  The condition on the size of the arc set in Property 4 is required since one can still apply RM moves to untangle a diagram which has only one arc, but there is no more equational rewriting to perform.

Assume that the above four properties are satisfied. Then, in any untangling sequence of diagrams the last diagram $D^U$ is a trivial diagram of the unknot. Then $D^U$ has just one arc with a label $\tau$, say. Let $D_k$ be the last diagram in the sequence which has this property of having just one arc with label $\tau$. Then, by Property 4, the label $\rho$ of the last removed arc (in the diagram $D_{k-1}$ preceding $D_k$ in the sequence) is provably equal to $\tau$, that is, $E(D_{k-1})\vdash (\tau  = \rho)$, and $(\tau  = \rho) \in E(D_{k-1})$. Unwinding the process backwards (and formally applying induction) we obtain that the labels of all of the removed arcs are provably equal to each other, including all of the generators. Thus, the required $n-1$ pairwise equalities of the $n$ generators are derivable from $E(D^U)$, and the result follows. The details of the construction of the sequence $(E(D_{i}),L_{i})$ can be found in 
the Appendix.  
For any $I$ which is a proper subset of $\{1,2,3\}$ the proof follows the same route using the property that the construction of $(E(D_{i+1}),L_{i+1})$
depends only on $(E(D_{i}),L_{i})$ and a type of RM used to transform $D_{i}$ into $D_{i+1}$.


%
%
%
%


The approach can be used to investigate which of the Reidemeister moves are required in a proof of unknottedness.

\begin{proposition}~\label{prop:culprit_all_three}
Culprit unknot (see above) needs all three Reidemeister moves to untangle.
\end{proposition}
{\bf Proof} For $I = \{2,3\}, \{1,3\},\{1,2\}$ one can disprove $E_{iq}^{I} \vdash TRIV$ by finding countermodels by Mace4 automatically of sizes 2,3,4, respectively.


~ \newline
An interesting question: is it possible to make a simulation in the opposite direction, that is, to extract an untangling sequences of Reidemeister moves from equational proofs? Although we don't have a definite answer here, in some simple cases one can indeed extract the moves from the proofs.
We leave the development of systematic Reidemester move extraction 
procedures for future work.

\section{Conclusion}

We presented the basis for a new method for unknot detection, based on parallel application of theorem prover and (counter)-model finder. It appears interesting, in that it has different abilities to existing approaches. In particular, the countermodel finder is producing the smallest non-trivial homomorphic image of the involutory quandle of the knot; thus it is, in some sense, finding the smallest invariant which distinguishes it from the unknot. Furthermore, the approach lends itself to new avenues of research, such as the exploration of the correlations between the equational proofs provided by the theorem provers and the corresponding sequences of labelled diagrams in an unknotting sequence.
Thus, whilst we have provided some interesting examples of unknot detection, exploring the whole spectrum of unknots and the relative difficulty of their detection, in comparison with other methodologies, will be an interesting avenue to explore. Furthermore, developing any correspondences between quandle-labelled diagram transformations and the unknotting proofs produced may provide interesting insights, potentially leading to more advanced tailored reasoning strategies.

Another direction is to explore automated deduction approach using different knot invariants such as knot groups and (non-involutory) quandles. In terms  of groups,
unknotedness corresponds to \emph{commutativity}
%
%
deciding  which can also be reduced to the equational theorem (dis)proving.  We have early indications that using involutory quandles, as explored in this paper, might be more efficient than using knot groups; for example, for a torus knot $T_{3,5}$ disproving using involutory quandles took less than a second and produced a countermodel of size 15, whilst disproving commutativity of the group of $T_{3,5}$ did not finish in 500 seconds and no countermodels of size less than $120$ were found.


%

\appendix







\newpage
\section*{Appendix}
%


The table below shows the experimental data obtained so far for the knots in the standard knot tables of at most 10 crossings. For each case, we present the size of the minimal countermodel to unknottedness which is found, together with the time taken to find this countermodel. In the majority of these cases the size of the countermodel is the smallest prime divisor of the determinant of the knot. In a few cases, the process was manually terminated 
after a certain time had elapsed, as indicated. 

\begin{table}
\caption{Experimental data for the time taken to find the minimal countermodels for the knots in the standard knot tables of at most 10 crossings.}
\scalebox{0.6}
{

 \begin{tabular}{|c|c|c|c|c|c|c|c|c|c|c|c|c|c|c|c|c|c|}

\hline
$Knot$   &  $3_{1} $  &  $4_{1}$  & $5_{1}$  & $ 5_{2}$  & $ 6_{1}$  &  $ 6_{2}$  &  $6_{3}$  &  $7_{1}$ & $7_{2}$ & $7_{3}$ & $7_{4}$ & $7_{5}$ &
$7_{6}$ & $7_{7}$ & $8_{1}$ & $8_{2}$ &

\\


$Size$     &  3  &        5      &  5        &  7  & 3 &  11  & 13 & 7 & 11 & 13 & 3 & 17 & 19 & 3  &
13 &
  17 &

 \\

$Time$   &  0.0  & 0.01 & 0.01 & 0.0 & 0.0 & 0.01 & 0.01 &

  0.01 &
  0.03 &
  0.03 &
  0. 0 &
  0.08 &
  0.16 &
  0.01 &
 0.06 &
  0.11 &

\\ \hline \hline


%

$Knot$ &
 $8_{3}$ &
 $8_{4}$ &
 $8_{5}$  &
  $8_{6}$  &
  $8_{7}$  &
  $8_{8}$ &
  $8_{9}$ &
  $8_{10}$ &
$8_{11}$ &
$8_{12}$ &
$8_{13}$ &
$8_{14}$ &
$8_{15}$ &
$8_{16}$ &
$8_{17}$ &
$8_{18}$ &
  \\


$Size$ &
  17 &
 19 &
 3 &
 23 &
 23 &
 5 &
 5 &
 3 &
 3  &
 29 &
 29 &
31 &
 3 &
5 &
37 &
3 &
  \\


$Time$ &
   0.12 &
   0.17 &
  0.01 &
   0.41 &
    0.47 &
     0.01 &
    0.01 &
   0.01 &
 0.01 &
1.31 &
2.11 &
2.06 &
0.01 &
 0.01 &
37.66 &
0.03 &
 \\




\hline \hline
$Knot$ &
$8_{19}$  &
$8_{20}$ &
$8_{21}$ &
$9_{1}$ &
$9_{2}$ &
$9_{3}$ &
$9_{4}$  &
$9_{5}$ &
$9_{6}$ &
$9_{7}$ &
$9_{8}$ &
$9_{9}$ &
$9_{10}$ &
$9_{11}$  & 	
$9_{12}$ &
$9_{13}$ &

\\ 

$Size$ &
3 &
3 &
3 &
3  &
3 &
19 &
 3 &
23 &
3 &
29 &
31 &
31 &
3 &
3 &
5 &
 37 &

\\


$Time$ &
 0.00 &
0.01  &
 0.01 &
 0.00 &
 0.00 &
 0.12 &
0.00 &
 0.45 &
 0.00 &
1.53 &
 3.77 &
1.50 &
 0.00 &
 0.01 &
0.01 &
3.77 &

\\ \hline  \hline
$Knot$ &
$9_{14}$ &
$9_{15}$ &
$9_{16}$ &
$9_{17}$ & 	
$9_{18}$ &
$9_{19}$ &
$9_{20}$ &
$9_{21}$ &
$9_{22}$ &
$9_{23}$ &
$9_{24}$ &
$9_{25}$ &
$9_{26}$ &
$9_{27}$ &	
$9_{28}$ &
$9_{29}$ &

\\


$Size$ &
37 &
3 &
3 &
3 &
41 &
41 &
 41 &
 43 &
 15 &
 3 &
 3 &
 15 &
 47 &
 7 &
 3 &
 3 &

 \\
$Time$ &
4.28 &
0.00 &
0.01 &
0.00 &
5.97 &
12.0 &
122.73 &
20.83 &
0.11  &
0.00  &
0.00  &
0.11  &
13.20  &
0.00  &
0.03 &
0.00  &

\\ \hline


$Knot$ &
$9_{30}$ &
$9_{31}$ &
$9_{32}$ &
$9_{33}$ &
$9_{34}$ &
$9_{35}$ &
$9_{36}$ &
$9_{37}$ &
$9_{38}$  &
$9_{39}$ &
$9_{40}$ &
$9_{41}$ &
$9_{42}$  &
$9_{43}$  &
$9_{44}$ &
$9_{45}$ &

\\


$Size$ &
15 &
5 &
59 &
61  &
3 &
3 &
15 &
3 &
3 &
5 &
3 &
7 &
7  &
13  &
15 &
15 &
\\


$Time$ &
0.09 &
0.00 &
1498.44 &
670 &
0.00 &
0.00 &
0.11 &
0.03 &
0.05 &
0.01 &
0.01 &
0.01 &
0.01 &
0.06 &
0.08 &
0.08 &
\\ \hline \hline

$Knot$ &
$9_{46}$  &
$9_{47}$ &
$9_{48}$ &
$9_{49}$  &
$10_{1}$ &
$10_{2}$ &
$10_{3}$ &
$10_{4}$  &
$10_{5}$  &
$10_{6}$  &
$10_{7}$  &
$10_{8}$  &
$10_{9}$  &
$10_{10}$ &
$10_{11}$ &
$10_{12}$ &

\\


$Size$ &
3 &
3 &
3 &
5 &
17 &
23  &
5 &
3 &
3 &
37 &
43 &
29 &
3 &
3 &
43 &
47 &

\\


$Time$ &
0.03 &
0.01 &
0.03 &
0.05 &
0.09 &
0.72 &
0.03 &
0.01 &
0.03 &
8.38 &
16.94 &

22.52 &
0.01 &
0.01 &
554  &
21.50 &

\\ \hline \hline
$Knot$ &
$10_{13}$ &	
$10_{14}$ &
$10_{15}$	&
$10_{16}$	&
$10_{17}$ 	&
$10_{18}$	&
$10_{19}$ 	&
$10_{20}$ 	&
$10_{21}$	&
$10_{22}$ 	&
$10_{23}$ 	&
$10_{24}$	&
$10_{25}$	&
$10_{26}$	&
$10_{27}$	&
$10_{28}$	&

\\

$Size$ &
53 &
3 &
43 &
47 &
41 &
5 &
3 &
5 &
3 &
7 &
59 &
5 &
5 &
61 &
71 &
53 &

\\

$Time$ &
90.53 &
0.03 &
16.22 &
11.84 &
18.14 &
0.01 &
0.03 &
0.00  &
0.00 &
0.01  &
144.42 &
0.01 &
0.00 &
1000 &
153 &
73.41 &

\\ \hline \hline

$Knot$ &
$10_{29}$ & 	
$10_{30}$ & 	
$10_{31}$ & 	
$10_{32}$ &	
$10_{33}$ & 	
$10_{34}$ & 	
$10_{35}$ & 	
$10_{36}$ & 	
$10_{37}$ & 	
$10_{38}$ & 	
$10_{39}$ & 	
$10_{40}$ & 	
$10_{41}$ & 	
$10_{42}$ & 	
$10_{43}$ & 	
$10_{44}$ & 	

\\

$Size$ &
3 &
67 &
3 &
3 &
5 &
37 &
7 &
3 &
53 &
59 &
61 &
3 &
71 &
3 &
73 &
79 &
\\


$Time$ &
0.01 &
298.11 &
0.03 &
0.00 &
0.03 &
11.38 &
0.03 &
0.03 &
36.98 &
157.08 &
76.25 &
0.00 &
554.25 &
0.01 &
728.36 &
968 &

\\ \hline


$Knots$ &
$10_{45}$ 	&
$10_{46}$	&
$10_{47}$	&
$10_{48}$	&
$10_{49}$	&
$10_{50}$	&
$10_{51}$	&
$10_{52}$	&
$10_{53}$	&
$10_{54}$	&
$10_{55}$ 	&
$10_{56}$	&
$10_{57}$	&
$10_{58}$	&
$10_{59}$	&
$10_{60}$	&

\\


$Size$ &
89 &
15 &
15 &
7 &
15 &
28 &
28 &
28 &
28 &
28 &
28 &
5 &

28 &
5 &
3 &
5 &
\\


$Time$ &
2629.41 &
0.11  &
0.11 &
0.01 &
0.14 &
4.12 &
6.33 &
4.76 &
5.94 &
3.28 &
8.53 &
0.03 &
5.51 &
0.01 &
0.03 &
0.01 &

\\ \hline \hline

$Knot$ &
$10_{61}$	&
$10_{62}$	&
$10_{63}$ 	&
$10_{64}$	&
$10_{65}$	&
$10_{66}$	&
$10_{67}$	&
$10_{68}$	&
$10_{69}$	&
$10_{70}$	&
$10_{71}$	&
$10_{72}$	&
$10_{73}$	&
$10_{74}$	&
$10_{75}$	&
$10_{76}$	&

\\


$Size$ &
3 &
3 &
3 &
3 &
3 &
3 &
3 &
3 &
3 &
15 &
7 &
15 &
15 &
3 &
3 &
3 &

\\


$Time$ &
0.01 &
0.01 &
0.03 &
0.01 &
0.03 &
0.01 &
0.01 &
0.01 &
0.01 &
 0.09 &
0.00 &
 0.12 &
 0.09 &
0.03 &
0.01 &
0.00 &

\\  \hline \hline

$Knots$ &
$10_{77}$	&
$10_{78}$	&
$10_{79}$	 &
$10_{80}$	&
$10_{81}$	&
$10_{82}$	&
${\bf 10_{83}}$	&
$10_{84}$	&
$10_{85}$	&
$10_{86}$	&
$10_{87}$ 	&
$10_{88}$ 	&
$10_{89}$	&
$10_{90}$	&
${\bf 10_{91}}$  &
${\bf 10_{92}}$  &

\\

$Size$ &
3  &
3 &
15 &
15 &
5 &
3 &
${\bf \ge 42}$ &
3 &
3 &
5  &
3 &
101  &
3 &
7 &
$ {\bf \ge 34}$  &
$ {\bf \ge 62} $ &

\\

$Time$ &
0.03 &
0.01 &
0.12 &
0.14 &
0.00 &
0.01 &
${\bf \ge 21170}$ &
0.00 &
0.00 &
 0.00 &
0.03 &
63849  &
0.00 &
0.03 &
$ {\bf\ge  20828}$ &
${\bf \ge 62848}$ &

\\ \hline \hline

$Knot$ &
$10_{93}$ 	&
$10_{94}$	&
$10_{95}$	&
$10_{96}$	&
$10_{97}$ 	&
$10_{98}$	 &
$10_{99}$	&
$10_{100}$	 &
$10_{101}$	&
$10_{102}$	&
$10_{103}$	&
$10_{104}$ 	&
$10_{105}$	&
$10_{106}$	&
$10_{107}$	&
$10_{108}$	&

\\

$Size$ &
15 &
71 &
7 &
3 &
3 &
3 &
3 &
5 &
5 &
15 &
3 &
7 &
7 &
3 &
3 &
3 &

\\

$Time$ &
0.62 &
40879.72 &
0.01 &
0.00  &
 0.01 &
0.01 &
0.01 &
0.01 &
0.01 &
1.01 &
0.00 &
 0.01 &
0.03 &
0.03 &
0.00 &
0.00 &

\\
\hline \hline

$Knot$ &
$10_{111}$	&
$10_{112}$  &
$10_{113}$	&
$10_{114}$	&
$10_{115}$	&
$10_{116}$ 	&
${\bf 10_{117}}$ 	&
$10_{118}$ 	&
${\bf 10_{119}}$	&
$ 10_{120}$ 	&
$10_{121}$ 	 &
$10_{122}$	 &
$10_{123}$ 	&
$10_{124}$ 	&
$10_{125}$ 	&
$10_{126}$	&

\\

$Size$ &
7 &
3 &
3  &
3  &
31 &
5 &
${\bf \ge 37}$  &
32 &
${\bf\ge 41}$  &
3 &
5 &
3 &
11 &
15 &
11 &
15 &

\\

$Time$ &
0.00 &
0.00 &
 0.00 &
0.03 &
 85098 &
0.01 &
${\bf\ge 150410}$ &
28 &
${\bf\ge  261100}$ &
0.00 &
0.00 &
0.03 &
 0.12 &
0.08 &
0.03 &
0.11 &

\\ \hline \hline

$Knot$ &
$10_{127}$	&
$10_{128}$ 	&
$10_{129}$	&
$10_{130}$ 	&
$10_{131}$	&
$10_{132}$ 	&
$10_{133}$ 	&
$10_{134}$ 	&
$10_{135}$ 	&
$10_{136}$	&
$10_{137}$	&
$10_{138}$	 &
$10_{139}$	&
$10_{140}$	&
$10_{141}$	&
$10_{142}$	&

\\

$Size$ &
15 &
11 &
5 &
17  &
28 &
5 &
19 &
23 &
28 &
3 &
5 &
 5 &
3 &
3 &
3 &
3 &

\\

$Time$ &
0.11  &
0.03  &
0.01 &
 0.28 &
6.5 &
0.03 &
 0.62 &
1.50 &
3.80 &
0.03 &
 0.00 &
0.03 &
0.01 &
0.03 &
0.03 &
0.01 &

\\  \hline  \hline

$Knots$ &
$10_{143}$	&
$10_{144}$	&
$10_{145}$	&
$10_{146}$	&
$10_{147}$	&
$10_{148}$	&
$10_{149}$ 	&
$10_{150}$	 &
$10_{151}$	 &
$10_{152}$	&
$10_{153}$	&
$10_{154}$	&
$10_{155}$	&
$10_{156}$	&
$10_{157}$ 	&
$10_{158}$	&

\\

$Size$ &
3 &
3 &
3 &
3 &
3 &
15 &
15 &
28 &
28 &
11 &
15 &
13 &
5 &
5  &
7 &
3 &
 \\

$Time$ &
0.00 &
0.01 &
0.01 &
0.01 &
0.00 &
0.12 &
0.11 &
6.72 &
6.36 &
0.06  &
0.12  &
0.08  &
0.00  &
0.01  &
0.03  &
0.01  &

\\  \hline \hline

$Knot$ &
$10_{159}$	&
$10_{160}$	&
$10_{161}$	 &
$10_{162}$	 &
$10_{163}$	&
$10_{164}$	&
$10_{165}$	 & & & & & & & &  & &

\\

$Size$ &
3 &
3 &
5 &
5 &
3 &
3 &
3 & & & & & & & & & &

\\
$Time$  &
0.01 &
0.03  &
 0.01  &
 0.03  &
0.03 &
0.01 &
0.01 & & & & & &  & & & &

\\  \hline

\end{tabular}

}

\end{table}


\subsection{Proof of Unknotedness of Culprit Unknot}
{\tiny
\begin{verbatim}
============================== prooftrans ============================
Prover9 (32) version Dec-2007, Dec 2007.
Process 2448 was started by Alexei on Alexei-PC,
Sun Jun  9 13:41:59 2013
The command was "/cygdrive/c/Program Files (x86)/Prover9-Mace4/bin-win32/prover9".
============================== end of head ===========================

============================== end of input ==========================

============================== PROOF =================================

% -------- Comments from original proof --------
% Proof 1 at 0.03 (+ 0.03) seconds.
% Length of proof is 145.
% Level of proof is 22.
% Maximum clause weight is 27.
% Given clauses 56.

1 a1 = a2 & a2 = a3 & a3 = a4 & a4 = a5 & a5 = a6 & a6 = a7 & a7 = a8 & a8 = a9 & a9 = a10 # label(non_clause) # label(goal).  [goal].
3 x * x = x.  [assumption].
4 (x * y) * y = x.  [assumption].
5 (x * y) * (z * y) = (x * z) * y.  [assumption].
6 a1 = a9 * a7.  [assumption].
7 a9 * a7 = a1.  [copy(6),flip(a)].
8 a3 = a1 * a2.  [assumption].
9 a1 * a2 = a3.  [copy(8),flip(a)].
10 a2 = a3 * a4.  [assumption].
11 a3 * a4 = a2.  [copy(10),flip(a)].
12 a5 = a2 * a10.  [assumption].
13 a2 * a10 = a5.  [copy(12),flip(a)].
14 a6 = a5 * a4.  [assumption].
15 a5 * a4 = a6.  [copy(14),flip(a)].
16 a7 = a6 * a1.  [assumption].
17 a6 * a1 = a7.  [copy(16),flip(a)].
18 a8 = a7 * a4.  [assumption].
19 a7 * a4 = a8.  [copy(18),flip(a)].
20 a10 = a8 * a9.  [assumption].
21 a8 * a9 = a10.  [copy(20),flip(a)].
22 a4 = a10 * a3.  [assumption].
23 a10 * a3 = a4.  [copy(22),flip(a)].
24 a9 = a4 * a8.  [assumption].
25 a4 * a8 = a9.  [copy(24),flip(a)].
26 a2 != a1 | a2 != a3 | a4 != a3 | a5 != a4 | a6 != a5 | a6 != a7 | a8 != a7 | a8 != a9 | a10 != a9.  [deny(1)].
27 a2 != a1 | a3 != a2 | a3 != a4 | a5 != a4 | a6 != a5 | a6 != a7 | a8 != a7 | a9 != a8 | a9 != a10.  [copy(26),flip(b),flip(c),flip(h),flip(i)].
29 (x * y) * x = x * (y * x).  [para(3(a,1),5(a,1,1)),flip(a)].
30 ((x * y) * z) * y = x * (z * y).  [para(4(a,1),5(a,1,1)),flip(a)].
31 (x * (y * z)) * z = (x * z) * y.  [para(4(a,1),5(a,1,2)),flip(a)].
32 ((x * y) * z) * (u * (y * z)) = ((x * z) * u) * (y * z).  [para(5(a,1),5(a,1,1))].
33 (x * (y * z)) * ((u * y) * z) = (x * (u * z)) * (y * z).  [para(5(a,1),5(a,1,2))].
34 a1 * a7 = a9.  [para(7(a,1),4(a,1,1))].
35 (a9 * x) * a7 = a1 * (x * a7).  [para(7(a,1),5(a,1,1)),flip(a)].
36 (x * a9) * a7 = (x * a7) * a1.  [para(7(a,1),5(a,1,2)),flip(a)].
37 a3 * a2 = a1.  [para(9(a,1),4(a,1,1))].
39 (x * a2) * a3 = (x * a1) * a2.  [para(9(a,1),5(a,1,2))].
40 a2 * a4 = a3.  [para(11(a,1),4(a,1,1))].
41 (a3 * x) * a4 = a2 * (x * a4).  [para(11(a,1),5(a,1,1)),flip(a)].
42 (x * a3) * a4 = (x * a4) * a2.  [para(11(a,1),5(a,1,2)),flip(a)].
46 a6 * a4 = a5.  [para(15(a,1),4(a,1,1))].
48 (x * a5) * a4 = (x * a4) * a6.  [para(15(a,1),5(a,1,2)),flip(a)].
49 a7 * a1 = a6.  [para(17(a,1),4(a,1,1))].
52 a8 * a4 = a7.  [para(19(a,1),4(a,1,1))].
55 a10 * a9 = a8.  [para(21(a,1),4(a,1,1))].
57 (x * a9) * a10 = (x * a8) * a9.  [para(21(a,1),5(a,1,2))].
58 a4 * a3 = a10.  [para(23(a,1),4(a,1,1))].
60 (x * a10) * a3 = (x * a4) * a2.  [para(23(a,1),5(a,1,2)),rewrite([42(4)]),flip(a)].
61 a9 * a8 = a4.  [para(25(a,1),4(a,1,1))].
63 (x * a8) * a9 = (x * a4) * a8.  [para(25(a,1),5(a,1,2))].
64 (x * a9) * a10 = (x * a4) * a8.  [back_rewrite(57),rewrite([63(8)])].
69 a9 * (a7 * a9) = a1 * a9.  [para(7(a,1),29(a,1,1)),flip(a)].
71 a3 * a10 = a2 * a3.  [para(11(a,1),29(a,1,1)),rewrite([58(7)]),flip(a)].
72 a10 * a4 = a4 * a2.  [para(11(a,1),29(a,2,2)),rewrite([58(3)])].
73 a2 * (a10 * a2) = a5 * a2.  [para(13(a,1),29(a,1,1)),flip(a)].
77 a10 * a8 = a7.  [para(21(a,1),29(a,1,1)),rewrite([61(7),52(6)])].
78 a9 * a10 = a4 * a9.  [para(21(a,1),29(a,2,2)),rewrite([61(3)]),flip(a)].
79 a10 * (a2 * a3) = a4 * a10.  [para(23(a,1),29(a,1,1)),rewrite([71(7)]),flip(a)].
80 a9 * a4 = a4 * a7.  [para(25(a,1),29(a,1,1)),rewrite([52(7)])].
84 (x * a7) * a9 = (x * a1) * a7.  [para(34(a,1),5(a,1,2))].
89 a3 * (a2 * a3) = a1 * a3.  [para(37(a,1),29(a,1,1)),flip(a)].
92 a2 * (a4 * a2) = a1.  [para(40(a,1),29(a,1,1)),rewrite([37(3)]),flip(a)].
104 ((x * a9) * a4) * a8 = x * a10.  [para(21(a,1),30(a,2,2)),rewrite([63(6)])].
114 (x * a7) * a1 = (x * a1) * a6.  [para(49(a,1),5(a,1,2)),flip(a)].
115 (x * a9) * a7 = (x * a1) * a6.  [back_rewrite(36),rewrite([114(8)])].
116 (x * a8) * a4 = (x * a4) * a7.  [para(52(a,1),5(a,1,2)),flip(a)].
117 (x * a9) * a8 = (x * a10) * a9.  [para(55(a,1),5(a,1,2))].
118 a10 * (a4 * a9) = a8 * a10.  [para(55(a,1),29(a,1,1)),rewrite([78(7)]),flip(a)].
121 (x * a3) * a10 = (x * a4) * a3.  [para(23(a,1),31(a,1,1,2)),flip(a)].
122 (x * a10) * a9 = (x * a4) * a7.  [para(25(a,1),31(a,1,1,2)),rewrite([117(4),116(8)])].
128 a10 * (a8 * a10) = a7 * a10.  [para(77(a,1),29(a,1,1)),flip(a)].
146 ((x * a4) * a8) * (y * a10) = ((x * a9) * y) * a10.  [para(21(a,1),32(a,1,2,2)),rewrite([63(4),21(13)])].
155 ((x * a1) * a7) * (y * a9) = ((x * a7) * y) * a9.  [para(34(a,1),32(a,1,2,2)),rewrite([34(13)])].
167 ((x * a4) * a3) * (y * a10) = ((x * a3) * y) * a10.  [para(58(a,1),32(a,1,2,2)),rewrite([58(13)])].
176 (a2 * a3) * (x * a10) = (a3 * x) * a10.  [para(71(a,1),5(a,1,1))].
177 (x * a10) * (a2 * a3) = (x * a4) * a3.  [para(71(a,1),5(a,1,2)),rewrite([121(10)])].
178 ((a2 * a3) * x) * a10 = a3 * (x * a10).  [para(71(a,1),30(a,1,1,1))].
180 (x * (a2 * a3)) * a10 = (x * a4) * a2.  [para(71(a,1),31(a,1,1,2)),rewrite([60(10)])].
191 (x * a10) * (a4 * a9) = (x * a4) * a8.  [para(78(a,1),5(a,1,2)),rewrite([64(10)])].
192 ((a4 * a9) * x) * a10 = a9 * (x * a10).  [para(78(a,1),30(a,1,1,1))].
194 (x * (a4 * a9)) * a10 = (x * a4) * a7.  [para(78(a,1),31(a,1,1,2)),rewrite([122(10)])].
217 (x * a10) * ((y * a4) * a8) = (x * (y * a9)) * a10.  [para(21(a,1),33(a,1,1,2)),rewrite([63(6),21(13)])].
241 (x * a10) * ((y * a4) * a3) = (x * (y * a3)) * a10.  [para(58(a,1),33(a,1,1,2)),rewrite([58(13)])].
288 a1 * (a4 * a2) = a2.  [para(92(a,1),4(a,1,1))].
291 (a4 * a2) * a1 = a4 * (a4 * a2).  [para(92(a,1),29(a,2,2)),rewrite([4(5)]),flip(a)].
297 ((x * a4) * a2) * a1 = x * (a4 * a2).  [para(92(a,1),32(a,1,2)),rewrite([4(10)])].
300 a1 * (a4 * (a4 * a2)) = a2 * a1.  [para(288(a,1),29(a,1,1)),rewrite([291(9)]),flip(a)].
310 (a4 * a1) * a6 = a1 * (a10 * a7).  [para(78(a,1),35(a,1,1)),rewrite([115(5)])].
311 a1 * (a4 * a7) = a4.  [para(80(a,1),35(a,1,1)),rewrite([4(5)]),flip(a)].
314 a4 * (a4 * a7) = a1.  [para(311(a,1),4(a,1,1))].
315 (a1 * x) * (a4 * a7) = a4 * (x * (a4 * a7)).  [para(311(a,1),5(a,1,1)),flip(a)].
318 a4 * ((a4 * a2) * a7) = a9.  [para(311(a,1),29(a,2,2)),rewrite([114(5),310(5),315(9),5(8),72(4),29(12),19(11),25(10)])].
326 a1 * a4 = a4 * a9.  [para(314(a,1),29(a,1,1)),rewrite([29(9),19(8),25(7)])].
336 (a4 * a9) * a1 = a1 * (a4 * a1).  [para(326(a,1),29(a,1,1))].
338 (x * (a4 * a9)) * a4 = (x * a4) * a1.  [para(326(a,1),31(a,1,1,2))].
345 a7 * (a7 * a9) = a6 * a9.  [para(69(a,1),29(a,2,2)),rewrite([4(5),5(12),49(8)])].
350 ((x * a1) * a6) * a9 = x * (a7 * a9).  [para(69(a,1),32(a,1,2)),rewrite([84(4),155(8),114(4),4(10)])].
372 (a10 * a5) * a2 = a10 * (a10 * a2).  [para(73(a,1),29(a,2,2)),rewrite([4(5),5(12)]),flip(a)].
376 ((x * a10) * a5) * a2 = x * (a10 * a2).  [para(73(a,1),32(a,1,2)),rewrite([5(8),4(10)])].
417 (x * (a2 * a3)) * (a4 * a10) = (x * a4) * a3.  [para(79(a,1),5(a,1,2)),rewrite([177(14)])].
418 a1 * a3 = a5.  [para(79(a,1),29(a,2,2)),rewrite([121(5),40(3),3(3),89(5),176(10),11(6),13(6)])].
420 ((x * a4) * a2) * (a2 * a3) = x * (a4 * a10).  [para(79(a,1),30(a,2,2)),rewrite([180(6)])].
421 (a4 * a1) * a2 = a4 * a10.  [para(79(a,1),30(a,2)),rewrite([23(3),39(5)])].
429 a3 * (a2 * a3) = a5.  [back_rewrite(89),rewrite([418(8)])].
455 a2 * a3 = a4 * a9.  [para(37(a,1),41(a,1,1)),rewrite([326(3),40(7)]),flip(a)].
461 a3 * (a4 * a9) = a5.  [back_rewrite(429),rewrite([455(4)])].
468 ((x * a4) * a2) * (a4 * a9) = x * (a4 * a10).  [back_rewrite(420),rewrite([455(7)])].
470 (x * (a4 * a9)) * (a4 * a10) = (x * a4) * a3.  [back_rewrite(417),rewrite([455(3)])].
480 (x * a4) * a7 = (x * a4) * a2.  [back_rewrite(180),rewrite([455(3),194(6)])].
482 a9 * (x * a10) = a3 * (x * a10).  [back_rewrite(178),rewrite([455(3),192(6)])].
483 (x * a4) * a8 = (x * a4) * a3.  [back_rewrite(177),rewrite([455(5),191(6)])].
485 a8 * a10 = a4 * a10.  [back_rewrite(79),rewrite([455(4),118(5)])].
497 (x * (a4 * a9)) * a10 = (x * a4) * a2.  [back_rewrite(194),rewrite([480(10)])].
503 ((a4 * a9) * x) * a10 = a3 * (x * a10).  [back_rewrite(192),rewrite([482(10)])].
506 (x * (y * a9)) * a10 = (x * (y * a3)) * a10.  [back_rewrite(217),rewrite([483(6),241(7)]),flip(a)].
507 (x * a10) * (a4 * a9) = (x * a4) * a3.  [back_rewrite(191),rewrite([483(10)])].
509 ((x * a9) * y) * a10 = ((x * a3) * y) * a10.  [back_rewrite(146),rewrite([483(4),167(7)]),flip(a)].
510 ((x * a9) * a4) * a3 = x * a10.  [back_rewrite(104),rewrite([483(6)])].
515 a10 * (a4 * a10) = a7 * a10.  [back_rewrite(128),rewrite([485(4)])].
516 a10 * (a4 * a9) = a4 * a10.  [back_rewrite(118),rewrite([485(8)])].
517 (x * a4) * a2 = x.  [back_rewrite(497),rewrite([506(6),58(3),4(4)]),flip(a)].
518 a3 * (x * a10) = a10 * (x * a10).  [back_rewrite(503),rewrite([509(6),58(3),29(4)]),flip(a)].
529 x * (a4 * a9) = x * (a4 * a10).  [back_rewrite(468),rewrite([517(4)])].
531 x * (a4 * a2) = x * a1.  [back_rewrite(297),rewrite([517(4)]),flip(a)].
537 a7 * a10 = a4 * a10.  [back_rewrite(516),rewrite([529(5),515(5)])].
538 (x * a4) * a3 = (x * a4) * a10.  [back_rewrite(507),rewrite([529(6),5(6)]),flip(a)].
544 (x * a4) * a10 = x.  [back_rewrite(470),rewrite([529(4),4(8),538(4)]),flip(a)].
547 a4 * a10 = a5.  [back_rewrite(461),rewrite([529(5),518(5),515(5),537(3)])].
550 (x * a4) * a6 = (x * a4) * a1.  [back_rewrite(338),rewrite([529(4),547(3),48(4)])].
554 a1 * (a4 * a1) = a2 * a1.  [back_rewrite(300),rewrite([531(6)])].
558 a2 = a1.  [back_rewrite(288),rewrite([531(5),3(3)]),flip(a)].
569 x * a9 = x * a10.  [back_rewrite(510),rewrite([538(6),544(6)])].
591 a5 = a4.  [back_rewrite(421),rewrite([558(4),4(5),547(4)]),flip(a)].
592 (x * a4) * a1 = x.  [back_rewrite(48),rewrite([591(1),4(4),550(4)]),flip(a)].
593 a4 * a1 = a1.  [back_rewrite(336),rewrite([569(3),547(3),591(1),554(8),558(4),3(6)])].
617 x * (a10 * a1) = x * a10.  [back_rewrite(376),rewrite([591(3),558(5),592(6),558(4)]),flip(a)].
618 a10 = a1.  [back_rewrite(372),rewrite([591(2),72(3),558(2),593(3),558(2),3(3),558(4),617(6),3(4)]),flip(a)].
629 a9 = a1.  [back_rewrite(318),rewrite([558(3),593(4),34(4),569(3),618(2),593(3)]),flip(a)].
648 a3 = a1.  [back_rewrite(40),rewrite([558(1),326(3),629(2),593(3)]),flip(a)].
649 a1 != a4 | a6 != a4 | a6 != a7 | a8 != a7 | a8 != a1.  [back_rewrite(27),rewrite([558(1),648(4),558(5),648(7),591(10),591(14),629(22),629(25),618(26)]),flip(h),xx(a),xx(b),xx(d),xx(i)].
650 a1 = a4.  [back_rewrite(13),rewrite([558(1),618(2),3(3),591(2)])].
653 x * a8 = x * a4.  [back_rewrite(350),rewrite([650(1),550(4),650(3),4(4),629(1),650(1),629(4),650(4),19(5)]),flip(a)].
654 a8 = a4.  [back_rewrite(345),rewrite([629(3),650(3),19(4),653(3),19(3),629(3),650(3),46(4),591(2)])].
681 a6 = a4.  [back_rewrite(15),rewrite([591(1),3(3)]),flip(a)].
687 a7 = a4.  [back_rewrite(77),rewrite([618(1),650(1),654(2),3(3)]),flip(a)].
688 $F.  [back_rewrite(649),rewrite([650(1),681(4),681(7),687(8),654(10),687(11),654(13),650(14)]),xx(a),xx(b),xx(c),xx(d),xx(e)].

============================== end of proof ==========================

\end{verbatim}
}

\subsection{To the proof of Proposition~\ref{prop:simulation}}

We explain now how to construct  a sequence $(E(D_{i}),L_{i})$.

\begin{itemize}
\item For the initial diagram $D=D_1$ we set $E(D) = R_{D}$ and $L_{1}$ assigns to every arc the corresponding generator from $G_{D}$;
\item If $D_{i+1}$ is obtained from $D_{i}$ by an application of $RM_{1}^{\downarrow}$ (the Reidemeister rule of type I, decreasing the number of crossings), as shown in Figure~\ref{RM1down}, then we set $E(D_{i+1}) = E(D_{i}) \cup \{ \tau = \rho \}$. The labelling function $L_{i}$ is updated to $L_{i+1}$ as shown in the figure (remaining the same outside the picture);
\begin{figure}
\hspace*{3cm}\includegraphics[scale=0.25]{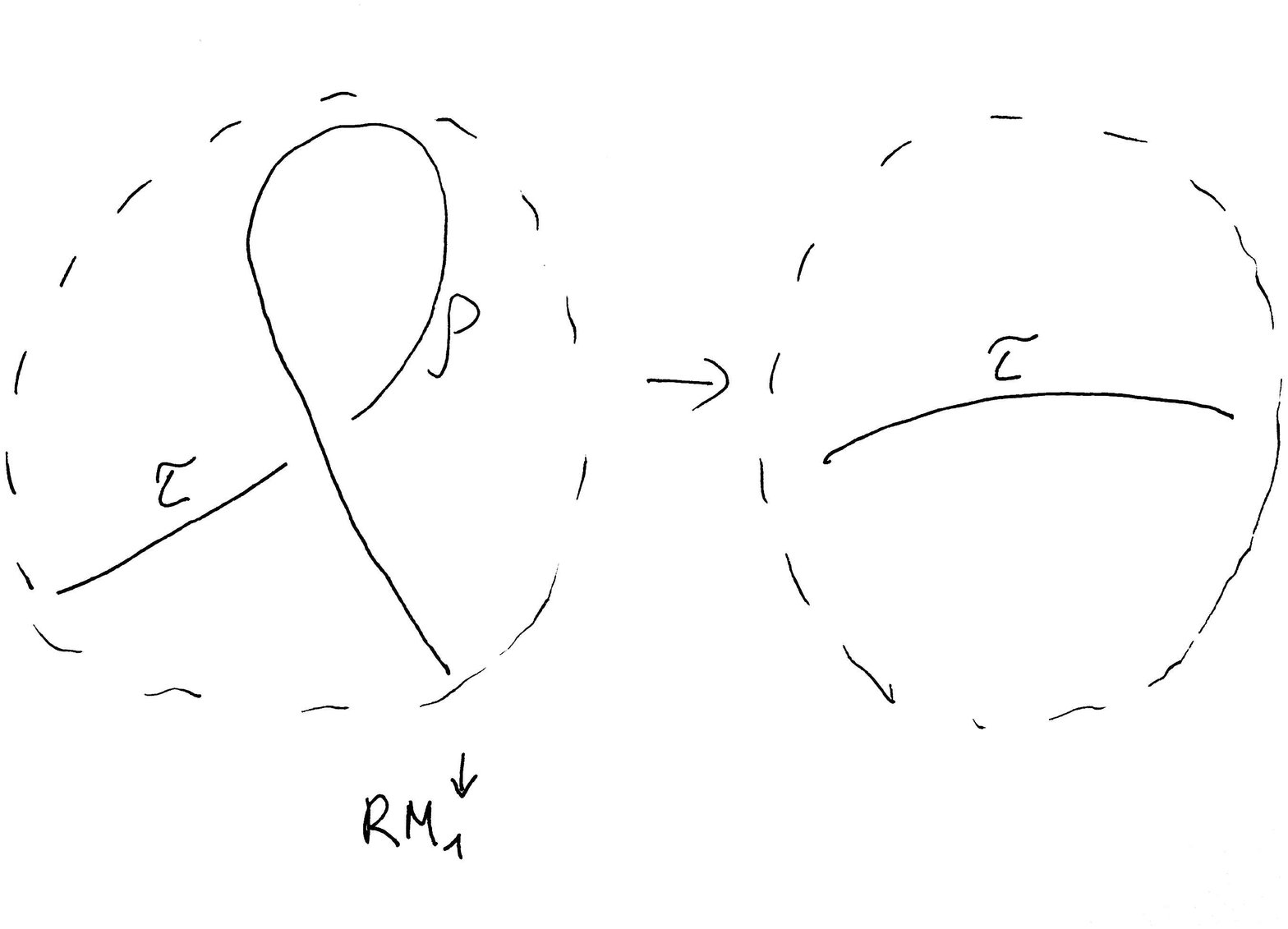}
\caption{$RM_{1}^{\downarrow}$}
\label{RM1down}
\end{figure}

\item If $D_{i+1}$ is obtained from $D_{i}$ by an application of $RM_{1}^{\uparrow}$ (the Reidemeister rule of type I, increasing the number of crossings) as shown in Figure~\ref{RM1up}, then a new arc labelled by $\rho$ is introduced;
$E(D_{i+1}) = E(D_{i}) \cup \{ \tau = \rho \}$. The labelling function $L_{i}$ is updated to $L_{i+1}$ as shown in the figure (remaining the same outside the picture);

\begin{figure}
\hspace*{3cm}\includegraphics[scale=0.25]{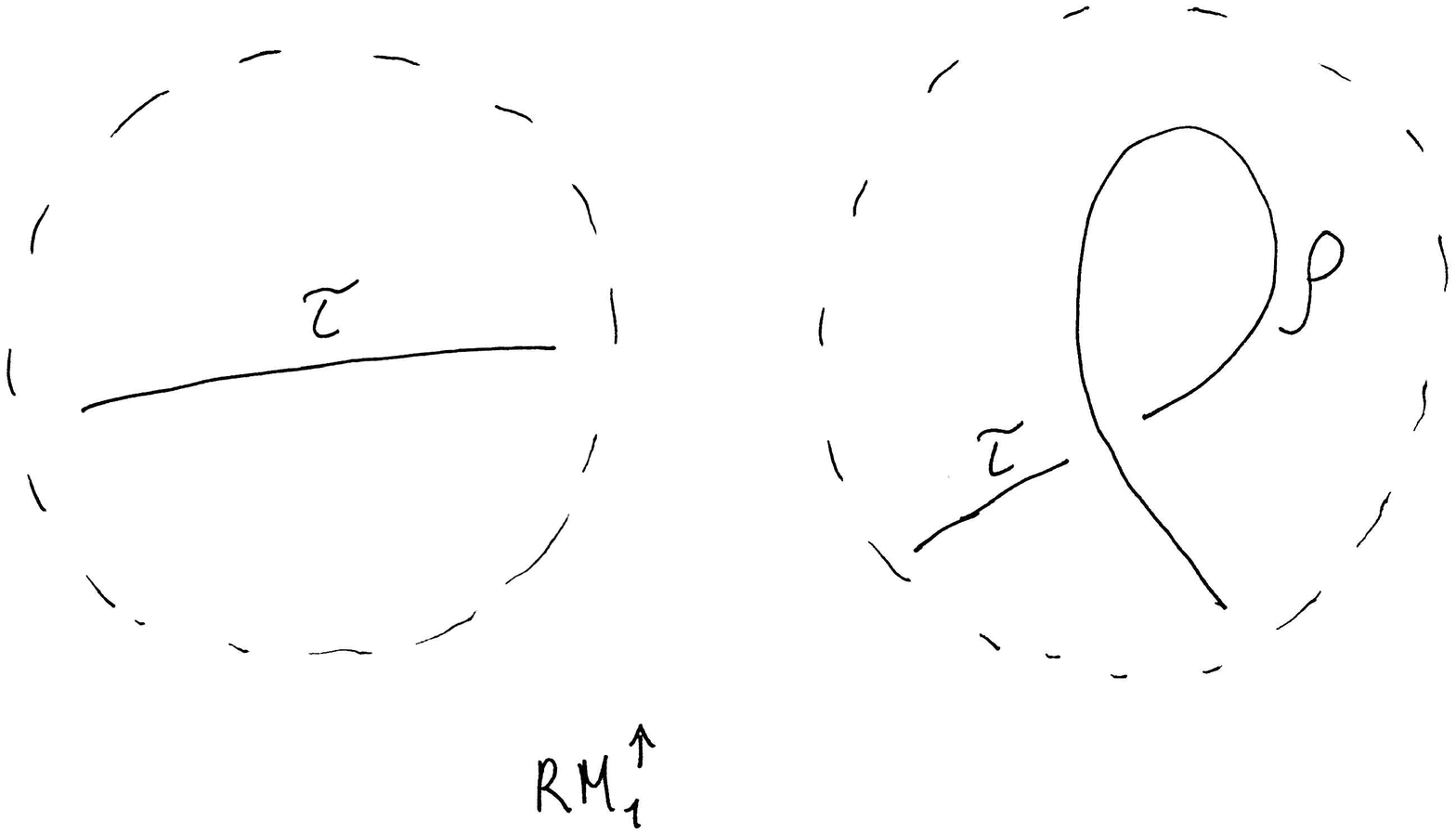}
\caption{$RM_{1}^{\uparrow}$}
\label{RM1up}
\end{figure}

 \item  If $D_{i+1}$ is obtained from $D_{i}$ by an application of $RM_{2}^{\uparrow}$ (the Reidemeister rule of type II, increasing the  number of crossings) as shown in Figure~\ref{RM2up}, then two new arcs labelled by $\theta$ and $\rho'$ are introduced;
$E(D_{i+1}) = E(D_{i}) \cup \{ \theta = \rho \triangleright \tau, \rho' = \theta \triangleright \tau \}$. The labelling function $L_{i}$ is updated to $L_{i+1}$ as shown in the figure (remaining the same outside the picture);

\begin{figure}
\hspace*{3cm}\includegraphics[scale=0.25]{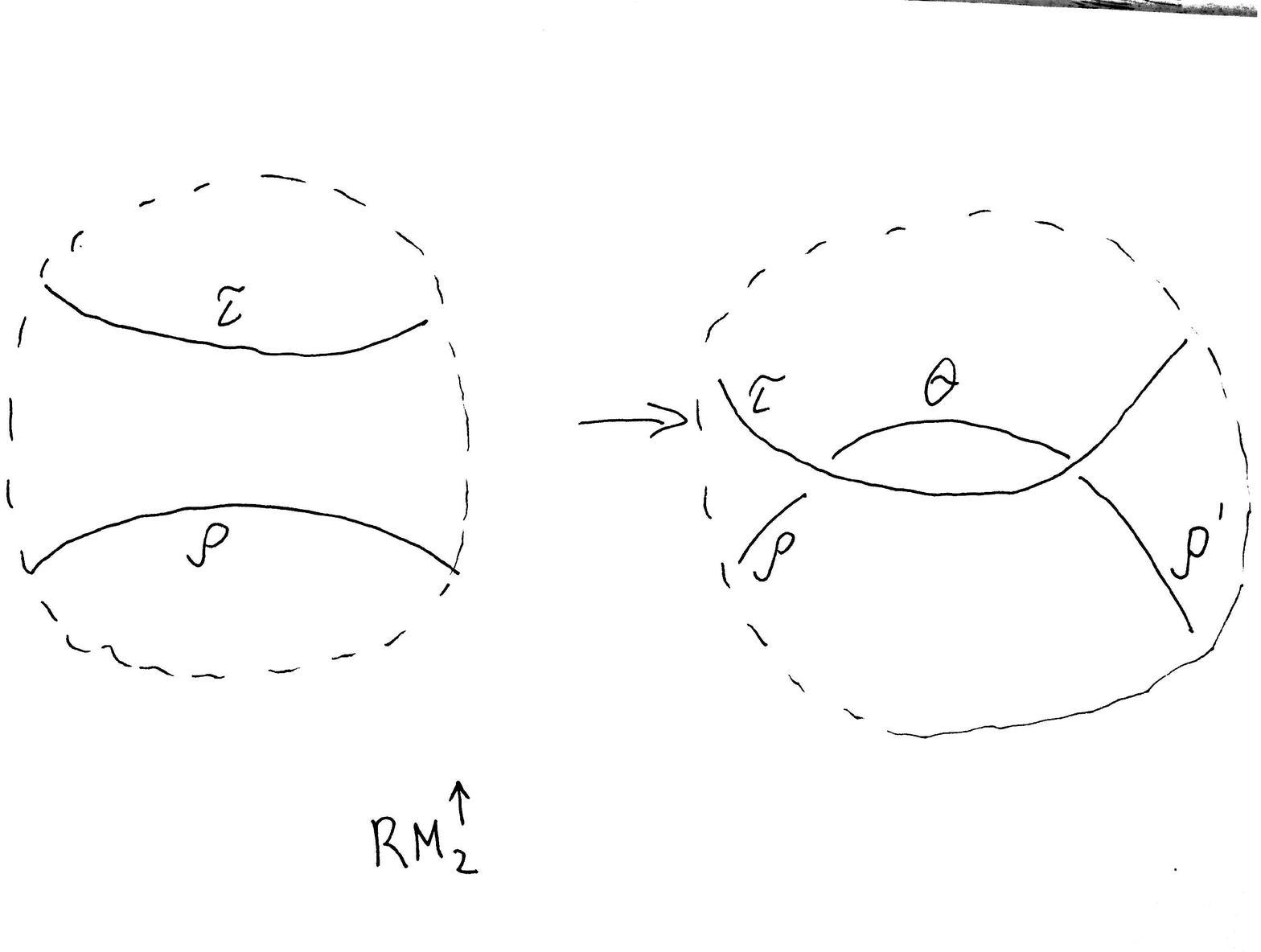}
\caption{$RM_{2}^{\uparrow}$}
\label{RM2up}
\end{figure}

\item   If $D_{i+1}$ is obtained from $D_{i}$ by an application of $RM_{2}^{\downarrow}$ (Reidemeister rule of type II decreasing the number of crossings) as shown in Figure~\ref{RM2down}, then $E(D_{i+1}) = E(D_{i}) \cup \{ \rho = \theta \}$. The labelling function $L_{i}$ is updated to $L_{i+1}$ as shown in the figure (remaining the same outside the picture);

\begin{figure}
\hspace*{3cm}\includegraphics[scale=0.25]{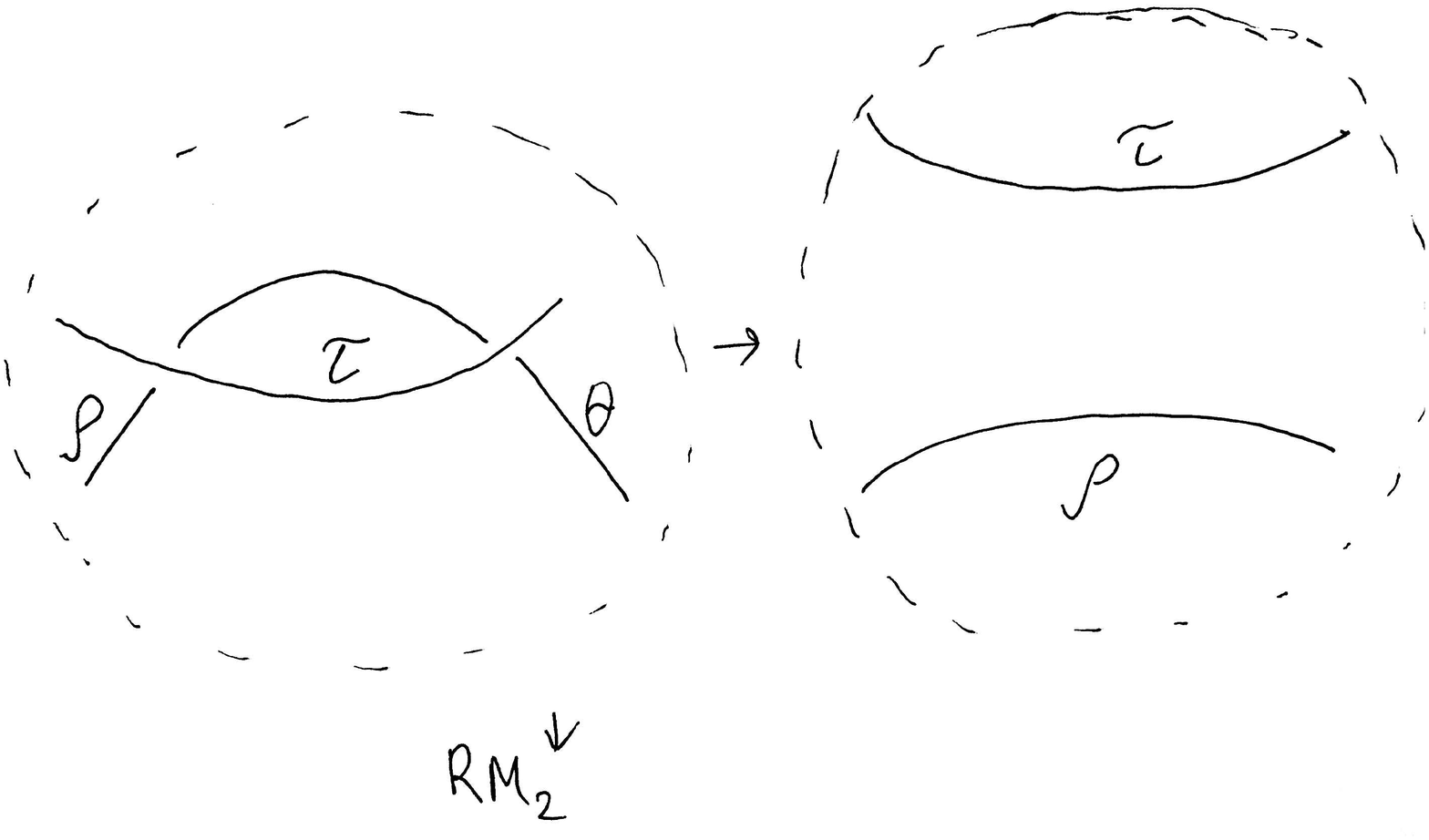}
\caption{$RM_{2}^{\downarrow}$}
\label{RM2down}
\end{figure}

 \item If $D_{i+1}$ is obtained from $D_{i}$ by an application of $RM_{3}$ (Reidemeister rule of type III) as shown in Figure~\ref{RM3}, then
$E(D_{i+1}) = E(D_{i}) \cup \{(\rho \triangleright \tau) \triangleright \theta = (\rho \triangleright \theta) \triangleright (\tau \triangleright \theta ) \}$. The labelling function $L_{i}$ is updated to $L_{i+1}$ as shown in the figure (remaining the same outside the picture).

\begin{figure}
\hspace*{3cm}\includegraphics[scale=0.25]{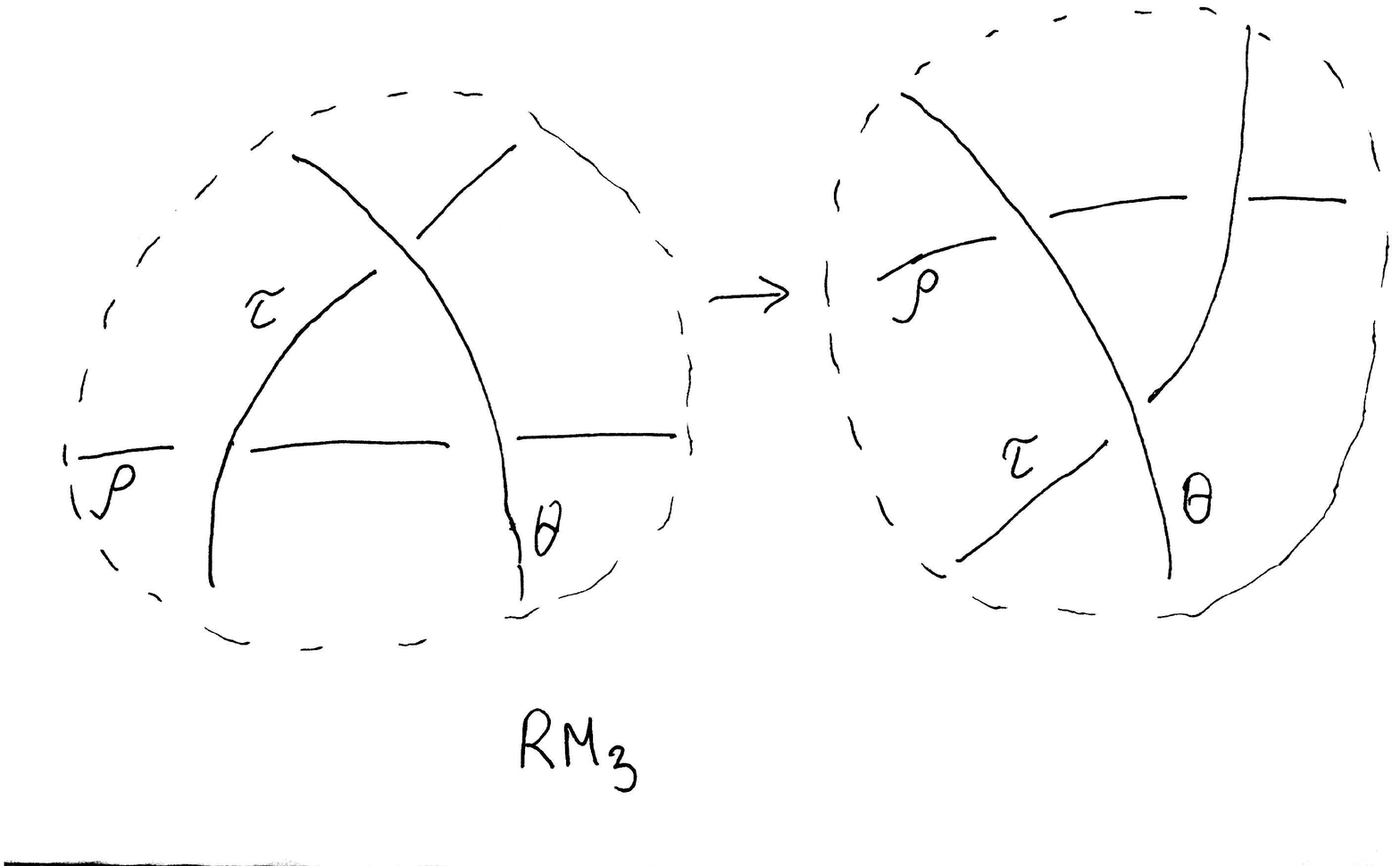}
\caption{$RM_{3}$}
\label{RM3}
\end{figure}

\end{itemize}

The required four properties can be checked by straightforward inspection.
$\Box$
$\Box$

\end{document}